\documentclass{article}

\usepackage{amsmath,amssymb}
\usepackage{cite}
\usepackage{fullpage}
\usepackage{amsthm}
\usepackage{graphicx}
\usepackage{microtype}
\usepackage{xcolor}
\usepackage{url}
\usepackage{enumitem}

\title{Support Optimality and Adaptive Cuckoo Filters\thanks{This work was supported in part by ISF grants no. 1278/16 and 1926/19, by a BSF grant no. 2018364, and by an ERC grant MPM under the EU’s Horizon 2020 Research and Innovation Programme (grant no. 683064).  A preliminary version of this paper appeared at the 17th Algorithms and Data Structures Symposium.}}

\author{Tsvi Kopelowitz\thanks{Bar-Ilan University, Ramat Gan, Israel. Email:~\protect\url{kopelot@gmail.com}}
	\and
Samuel McCauley\thanks{Williams College, Williamstown, MA 01267 USA. Email:~\protect\url{srm2@williams.edu}}
\and
Ely Porat\thanks{Bar-Ilan University, Ramat Gan, Israel. Email:~\protect\url{porately@cs.biu.ac.il}}
}

\date{}

\newtheorem{theorem}        {Theorem}

\newtheorem{lemma}[theorem] {Lemma}

\newtheorem{definition}[theorem]{Definition}

\newcommand{\defn}[1]{\textbf{\emph{#1}}}
\newcommand{\E}{\mathbb{E}}

\newcommand{\present}{\texttt{present}}
\newcommand{\absent}{\texttt{absent}}

\newcommand{\F}{\mathcal F}

\newcommand{\samnote}[1]{\textcolor{blue}{$\ll$\textsf{#1 --Sam}$\gg$\marginpar{\tiny\bf SM}}}
\newcommand{\sam}{\samnote}

\newcommand{\ouradaptive}{support optimal}
\newcommand{\ouradaptiveadj}{support-optimal}

\newif\iffull
\fulltrue

\allowdisplaybreaks

\begin{document}
\maketitle

\begin{abstract}
  Filters (such as Bloom Filters) are a fundamental data structure that speed up network routing and measurement operations by storing a compressed representation of a set.
  Filters are very space efficient, but can make bounded one-sided errors: with tunable probability $\epsilon$, they may report that a query element is stored in the filter when it is not.  This is called a \defn{false positive}.
  Recent research has focused on designing methods for \defn{dynamically adapting} filters to false positives, thereby reducing the number of false positives when some elements are queried repeatedly.

  Ideally, an adaptive filter would incur a false positive with bounded
  probability $\epsilon$ for each new query element, and 
  would incur $o(\epsilon)$ total false positives over all repeated queries to
  that element.  
  We call such a filter
  \defn{\ouradaptive{}}.

  In this paper we design a new Adaptive Cuckoo Filter, and show that it is \ouradaptive{} (up to additive logarithmic terms) over any $n$ queries when storing a set of size $n$.  
  \iffull
  Our filter is very simple: fixing previous false positives requires a simple cuckoo-like operation, and the filter does not need to store any additional metadata.  This data structure is the first practical data structure that is \ouradaptive{}, and the first \ouradaptive{} filter that does not require additional space beyond a normal cuckoo filter.
  \fi

  We complement these bounds with experiments that show that our data structure is effective at fixing false positives on network trace datasets, outperforming previous Adaptive Cuckoo Filters.

  Finally, we investigate adversarial adaptivity, a stronger notion of adaptivity in which 
  an adaptive adversary repeatedly queries the filter, using the result of previous queries to drive the false positive rate as high as possible.
  We prove a lower bound showing that a broad family of filters, including all known Adaptive Cuckoo Filters, can be forced by such an adversary to incur a large number of false positives.

\end{abstract}

\section{Introduction}
\label{sec:intro}

A \defn{filter} is a data structure that supports membership queries for a set of elements $S = x_1,\ldots x_n$ from a universe $U$.  The answer to each filter query is \present{} or \absent{}.
Typically, a filter has a \defn{correctness} guarantee: if an element $q\in S$, the filter must return \present{} to the query with probability 1.
There is also a \defn{performance} guarantee: if an element $q\notin S$, the filter must return \present{} with tunable probability at most $\epsilon$.  
If a query on an element $q\notin S$ returns \present{} then $q$ is called a \defn{false positive}.
Typically, filters use a small amount of space.

A filter's small size means that the filter can be stored in an efficiently accessible location.
Meanwhile, the no-false-negative guarantee implies that if the filter returns $q\notin S$ for a query $q$, then there is no need for accessing the actual data, which is typically stored in a medium with expensive access cost.
This ability to filter out queries to items not in $S$ in a small-size structure has found a wide variety of 
network applications such as collaboration in peer-to-peer networks, resource routing, packet routing, and measurement infrastructures~\cite{BroderMi04} as well as many areas of network security~\cite{GeravandAh13}.

There are several different kinds of filters.
The Bloom filter~\cite{Bloom70} was the first filter data structure to be designed; 
it is still very popular due to its simplicity and efficiency.
Later filters were designed to provide better worst-case lookup times and space guarantees~\cite{BenderFaGo18,PaghPaRa05,Porat09}, improved practical performance~\cite{FanAnKa14,WangZhSh19}, and improved cache performance~\cite{BenderFaJo12}.

In this paper, we focus on filters that achieve space very close to the optimal lower bound of $n\log 1/\epsilon$ bits~\cite{CarterFlGi78,LovettPorat10}, and that store elements from a large universe $|U| \gg n$.

\paragraph{Fixing False Positives}

A well-known issue with many existing filters is that they cannot \defn{adapt} to queries: if a query $q\notin S$ is a false positive, all subsequent queries $q' = q$ will be false positives.  
The focus of this paper is designing filters that do adapt to false positive queries, so that if a query $q$ is a false positive, the filter undergoes structural changes so that a later query to $q$ is unlikely to be a false positive.  
An element $q$ is said to be \defn{fixed} if $q$ was previously a false positive, but is no longer a false positive.  Similarly, $q$ is \defn{broken} if $q$ was previously fixed, but is now a false positive.

\paragraph{Related Work.}
Bender et al.~\cite{BenderFaGo18}  analyzed how to fix false positives against an adversary.
They give a data structure such that if queries are generated by an adversary trying to maximize the false positive rate, each query to a filter is a false positive with probability at most $\epsilon$, even if the query element was queried before.  
This requirement essentially provides concentration bounds: over $n$ queries, their filter incurs $\epsilon n$ false positives, even if the queries are maliciously chosen based on previous false positives.

However, the benefit of adaptivity goes beyond resisting an adversary.  As shown experimentally by Mitzenmacher et al.~\cite{MitzenmacherPo17}, adapting to queries can significantly decrease the number of false positives---in fact, if queries are repeated sufficiently frequently, the performance can be much better than $O(\epsilon n)$.  In particular, network trace data consists of a structured sequence of queries---can we give a data structure that performs particularly well on this kind of data?

Most recently, Bender et al.~\cite{BenderDaFa21} compared adaptivity to cache-based strategies, finding that adaptivity leads to significantly better practical performance.

\paragraph{Support Optimality}
Ideally, an adaptive filter would incur a false positive with probability $\epsilon$ for each new query, and incur no further queries asymptotically.   Thus, every new false positive is fixed, and this fixing is unlikely to break previously-fixed false positives.
In particular, let $q_1,\ldots q_n$ be a predetermined sequence of queries%
\footnote{Note that the filter does not have access to this sequence ahead of time; it must process the queries online.  We fix a predetermined sequence of queries to clarify that, unlike in~\cite{BenderFaGo18}, we do not allow each new query to be determined adversarially based on the result of previous queries.}
 to a filter $\mathcal{F}$, and let $Q = \bigcup_{i = 1}^n \{q_i\}$ be the set of unique queries in the sequence.  We say that $\mathcal{F}$ is \defn{\ouradaptive{}} if the expected number of false positives when querying $q_1,\ldots q_n$ is $\epsilon|Q|(1 + o(1))$.  
In this paper we give a \ouradaptiveadj{} filter up to additive polylogarithmic terms, and show that it significantly improves practical performance.

\iffull
\paragraph{Avoiding Remote Memory Accesses}
Adaptive filters focus on the case where a filter is used to avoid remote memory accesses.  Oftentimes, network data is stored on large, slow storage.  In this case, if the filter is small enough to fit in faster storage, it can be used to rule out queries to elements not in the set. Thus, when the filter answers \texttt{absent}, a slow query to the large storage is avoided, greatly improving throughput.  

This use case is important because it allows us to assume (costly) access to the original set of items being held in the filter.  Bender et al.~\cite{BenderFaGo18} gave a lower bound showing that this access is necessary to achieve adaptivity, and indeed this assumption was used both in their data structre as well as in the Adaptive Cuckoo Filter~\cite{MitzenmacherPo17}.

In particular, we assume that we can make $O(1)$ expected accesses to $S$ on each insert, as well as each time the filter answers \texttt{present}.  This does not asymptotically increase the number of accesses to slow storage.\footnote{In fact, if $S$ is stored to allow for efficient, in-place reverse hash lookups---so that the normal hash lookup and the adaptive rehashing can be performed simultaneously---then some false positives may be fixed without increasing the number of remote accesses at all.}  We use this assumption to ``cuckoo'' elements on a false positive, rehashing them and swapping them to another slot.  All previous adaptive cuckoo filters also require rehashing elements of $S$ on a false positive~\cite{MitzenmacherPo17,BenderFaGo18}.
\fi

\subsection{Results}

We discuss three data structures in this paper: two versions of the Adaptive Cuckoo Filter originally presented in~\cite{MitzenmacherPo17} (which we call the \defn{Cyclic ACF} and \defn{Swapping ACF}), and a Cuckoo Filter augmented with a new method of achieving adaptivity, which we call the \defn{Cuckooing ACF}.

The first contribution of this paper is the Cuckooing ACF, a \ouradaptiveadj{} filter which can be implemented using almost-trivial changes to current Cuckoo Filter implementations.  

In Section~\ref{sec:upper}, we analyze the Cuckooing ACF and prove that it is \ouradaptive{} over any $n$ queries, up to additive polylogarithmic terms.  
This gives a significant performance improvement over previous filters even for large $Q$, and the difference becomes more dramatic for small $Q$.   For example, for the case of a repeated single query ($|Q| = 1$), 
static filters incur $\epsilon n$ false positives in expectation, whereas we show that the Cuckooing ACF incurs $O(\log^4 n)$  expected false positives.  

We show that despite their strong practical performance, the Cyclic ACF and Swapping ACF are not \ouradaptive{}--even if there are a constant number of queries ($|Q| = O(1)$), they may incur $\Omega(n)$ false positives, whereas the Cuckooing ACF incurs at most $O(\log^4 n)$.  Thus, from the standpoint of \ouradaptive{}ity, cuckooing is a better method for achieving adaptivity.

In Section~\ref{sec:experiments}, we provide experimental results that show that the theory bears out in practice: the Cuckooing ACF attains a low false positive rate on network trace datasets, which contain many repeated queries.  The performance is not only stronger than a vanilla Cuckoo Filter, but also improves upon the performance of a Cyclic ACF and a Swapping ACF of the same size.  This shows that the Cuckooing ACF is effective at fixing false positives in a practical sense. These results also emphasize the benefit of a simple adaptive filter: not only is the resulting data structure easier to implement, the simplicity entails less space usage compared to previous Adaptive Cuckoo Filters, leading to a significant performance improvement. 

Finally, in Section~\ref{sec:lower}, we prove lower bounds that demonstrate that a broad family of filters cannot be adaptive in the adversarial sense of Bender et al.~\cite{BenderFaGo18}; this includes the Cyclic ACF, the Swapping ACF, and the Cuckooing ACF.  
This lower bound motivates the concept of \ouradaptive{}ity: a \ouradaptive{} filter achieves strong performance on real datasets without achieving adversarial adaptivity.  
Our proof also gives insight into the structure of adaptive filters---specifically, it shows that a space-efficient filter must have variable-sized fingerprints in order to be adversarially adaptive.

\iffull
\subsection{Related Work}
Mitzenmacher et al.~\cite{MitzenmacherPo17} were the first to describe the notion of adaptivity.  They show that Adaptive Cuckoo Filters can significantly improve the number of false positives incurred on real-world network trace datasets.  We replicate their experiments in Section~\ref{sec:experiments}, and show that the Cuckooing ACF improves upon their performance.  They also give theoretical bounds that guarantee that a small number of false positives are incurred by their data structures if queries are selected at random.  Their work briefly mentions the swapping strategy used in this paper (in the Cuckooing ACF), but does not analyze it.

Bender et al.~\cite{BenderFaGo18} give a filter (achieving optimal space and optimal query time) that is adaptive against an adversary.  Their filter works by adding a constant number of bits to the filter on average for each false positive encountered.  This data structure gives strong bounds, but dynamically changing the number of bits for each stored element is likely to be an obstacle to a practical implementation.  Bender et al. also give a lower bound showing that access to an external dictionary storing $S$ is necessary to achieve adaptivity; this helps motivate such accesses in our paper.
\fi

\iffull
The Bloomier filter of Chazelle et al.~\cite{ChazelleKiRu04} and the Weighted Bloom Filter of Bruck et al.~\cite{BruckGaJi06} also aim to reduce the false positive rate; however, both assume that frequent queries are known to the filter in advance.
\fi

\section{Three Adaptive Cuckoo Filters}
\label{sec:filters}

In this section describe a new kind of filter, the Cuckooing ACF\@.  
We then discuss the Cyclic ACF and the Swapping ACF, both originally introduced in~\cite{MitzenmacherPo17}. 

We include a symbol table for reference in Appendix~\ref{sec:symboltable}.

\subsection{ACF Parameters and Internal State}%
\label{sec:ACFprelim}

We begin by defining a more general data structure which we call the \defn{adaptive cuckoo filter} (ACF). As the name suggests, the Cyclic ACF, the Swapping ACF, and the Cuckooing ACF are adaptive cuckoo filters.

An ACF $\F$ has integer parameters $f, k, b > 0$, an additional parameter $\gamma > 1$, and supports storing $n$ elements from a universe $U$ with a false positive rate $\epsilon$.  
The internal representation of a filter $\F$ consists of $k$ hash tables, each of $N=\gamma n/bk$  bins,\footnote{We assume $\gamma n$ is an integer multiple of $bk$ for simplicity.} where each bin consists of $b$ slots of $f$ bits; 
thus, the space usage of $\F$ is $N\cdot b\cdot f\cdot k$ bits. 
The parameter $\gamma$ determines how densely elements are packed, trading off between insert time and space; often $\gamma \approx 1.05$ is used.

The hash tables are accessed using $k+1$ hash functions: 
 $k$ \defn{location hash} functions $h^{\ell}_1,\ldots, h^{\ell}_k: U \rightarrow \{0,\ldots, N-1\}$ that hash from $U$ to a bit string of length $\log N$,\footnote{When treating the hash value as a bit string we assume that $N$ is a power of two for simplicity; this assumption is not necessary for the implementation.}  and a single \defn{fingerprint hash} $h^f$ mapping each $x\in U$ to an $f$-bit \defn{fingerprint}.  
 The range and domain of $h^f$ depend on which ACF is used
 and may depend on the internal state of $\F$; we provide details below.
Following previous results on filters~\cite{Bloom70,BenderFaGo18,BenderFaJo12,FanAnKa14,Eppstein16,MitzenmacherPo17}, this paper assumes free access to uniform random hash functions.\footnote{While such strong hashes are not useable in practice, this analysis is generally predictive of experimental results (see i.e.~\cite{FanAnKa14,PandeyBeJo17,MitzenmacherPo17}).}

When a set $S$ is stored in $\F$, for each element $x_i\in S$, the fingerprint of $x_i$ is stored in one of the slots of bin $B(x_i)$ in the $\beta_i$th hash table; this bin is defined using a location hash: $B(x_i) =h_{\beta_i}^{\ell}(x)$ for some integer $0\le \beta_i < k$.
We say a slot $\sigma$ is \defn{occupied} if the fingerprint of some $x_i\in S$ is stored in $\sigma$; otherwise $\sigma$ is \defn{empty}.
We call $\beta_i$ the \defn{hash index} of $x_i$. 

Since an ACF stores each element using a hash index, we can keep track of the internal state of a filter using the hash index of each element.  
Thus,  we use $C= (C[1],C[2],\ldots,C[n]) = (\beta_{1}, \ldots, \beta_{n})$ to define the \defn{configuration} of $\F$.
This fully defines the internal representation of a Cuckooing ACF.  The internal representation of a Cyclic ACF also depends on $s$ metadata bits stored for each element, and the internal representation of a Swapping ACF also depends on which slot within the bin is used to store each element.

Suppose $S$ is stored using hash indices $\beta_1,\ldots \beta_n$.
Then  query $q\notin S$ \defn{collides} with an element $x_i\in S$ under configuration $C$ when
$h^{\ell}_{\beta_i}(x_i) = h^{\ell}_{\beta_i}(q)$ and
$q$ and $x_i$ have the same fingerprint.

 \subsection{Cuckoo Filter Operations}
 \label{sec:ACFoperations}

 We begin by describing how inserts and queries work for an ACF.  The Cuckooing ACF and Swapping ACF insert and query using these methods; the Cyclic ACF uses a generalization of these methods.

\paragraph{Insert.}
Suppose an element $x_{i}$ is inserted into a set $S$ of size $i-1$ currently stored with filter $\mathcal{F}$ in configuration $C$, where elements $S = x_1,\ldots x_{i-1}$ have hash indices $\beta_1,\ldots \beta_{i-1}$.  Assume that $\mathcal{F}$ can store up to $n \geq i$ elements.
The insertion algorithm finds a valid configuration $C'$ of $\F$ on $S$ such that there exists a hash index $\beta'_{i}\in \{0,\ldots k-1\}$ for which bin $h^{\ell}_{\beta_{i}'}$ has an empty slot. 
This may involve updating the hash indices of other elements; for $1\leq j < i$ let $\beta'_j$ be the hash index of $x_j'$ under $C'$.
We describe how to determine $C'$ below.

If there is already an available empty slot, the filter stores the element immediately in that slot.
Specifically, if there exists a $\beta\in\{0,\ldots,k-1\}$ where bin $h^{\ell}_{\beta}(x_i)$ in hash table $\beta$ contains an empty slot, the filter sets $\beta'_{i} = \beta$, and stores the fingerprint of $x_i$ in the empty slot. All other slots remain unchanged: $\beta_j' = \beta_j$ for all $1\leq j < i$.

Now, consider the case where there is no available empty slot.  Then the ACF makes room by shifting elements as one would in cuckoo hashing~\cite{PaghRodler04}.
The filter selects a hash index $\beta_{i}$ arbitrarily from $\{0,\ldots,k-1\}$.
Since all slots in bin $h^{\ell}_{\beta_{i}'}(x_{i})$ are occupied in $C$, the filter \defn{moves} the fingerprint of some element $x_j$ stored in a slot in $h^{\ell}_{\beta_{i}'}(x_j) = h^{\ell}_{\beta_{i}'}(x_i)$, leaving an empty slot in which $x_i$ can be stored.
If $h_\beta^{\ell}(x_j)$ contains an empty slot for some $\beta\in\{0,\ldots,k-1\}$
(i.e. if $x_j$ can be stored in an empty slot), one such empty slot is arbitrarily selected to store $x_j$.
Otherwise, the filter increments $\beta'_{j} = \beta_{j} + 1\pmod k$ and recurses, moving an element stored in $h^{\ell}_{\beta'_{j}}(x_j)$ as necessary.  

The move the elements as described above, the ACF must be able to access the set $S$ during an insert in order to rehash each $x_j$.
We follow all past work on adaptive filters~\cite{BenderFaGo18,MitzenmacherPo17,BenderDaFa21} in assuming that an external dictionary can be accessed, enabling an element to be rehashed while inserting or fixing.

 If this recursive process takes too many steps (more than $\Theta(\log n)$ elements are moved), the filter chooses new hashes and is rebuilt from scratch.
 If $\mathcal{F}$ uses $N = O(n)$ hash slots, then over $n$ inserts, the probability of a rebuild is $O(1/n)$~\cite{PaghRodler04}.

\iffull
There has been a great deal of previous work on improving the above cuckoo hashing procedure with various methods for finding the best slots to store each element, see i.e. \cite{KirschMiWi09,FotakisPaSa05,FriezeMeMi09,ArbitmanNaSe09,NaorSeWi08}.

Notice that the procedure we describe does not use ``partial-key'' cuckoo hashing, as in the original cuckoo filter of Fan et al.~\cite{FanAnKa14}.  The ability to access $S$ (necessary for adaptivity~\cite{BenderFaGo18}) means that elements can be moved based on an entirely new hash.
\fi

\paragraph{Query.}
On a query $q$, a filter $\F$
in configuration $C$ returns \present{} 
if there exists a $\beta$ 
and a slot index $\sigma \in \{1,\ldots b\}$ 
such that slot $\sigma$ in bin $h_{\beta}^\ell(q)$ of table $\beta$ is occupied and stores the fingerprint of $q$.
This immediately guarantees correctness of the filter (queries to $x_i\in S$ always return \present{}) and, via a union bound over the elements of $S$, a false positive rate of at most $n/(N2^f)$.
  The filter achieves a desired false positive rate $\epsilon$ by setting $f = \log (n/(N\epsilon)) = \log (bk/\epsilon\gamma)$.

\paragraph{Fixing false positives.}  If an ACF returns \present{} on a false positive query $q$ (the filter knows that $q\notin S$ from the external dictionary storing $S$), the ACF modifies its configuration 
to attempt to fix $q$, 
so that subsequent queries to $q$ return \absent{}.  Each type of ACF has its own method for fixing false positives, which we describe below.
Notice that the process of modifying the configuration may cause some query $q'\notin S$ to become a false positive, even if $q'$ was fixed some time in the past.

\subsection{Cuckooing ACF}
\label{sec:cuckooingACF}

The primary data structure contribution of this paper is the \defn{Cuckooing ACF}.  This data structure is a standard Cuckoo Filter~\cite{FanAnKa14} with an added operation to fix false positives; inserts and queries work exactly as described in Section~\ref{sec:ACFoperations}.

Let $q$ be a false positive under configuration $C$; we define how the Cuckooing ACF finds a new configuration $C'$ with hash indices $\beta_1', \ldots \beta_n'$ to attempt to fix $q$.
For each $x_i\in S$ that collides with $q$ under $C$, the filter moves $x_i$ recursively as it would during an insert.  Specifically, the filter 
sets the new hash index $\beta'_i = \beta_i + 1\pmod k$; 
if bin $h^\ell_{\beta'_i}(x_i)$ in table $\beta'_i$ does not contain an empty slot, 
an element $x_j$ stored in $h^{\ell}_{\beta'_i}(x_i)$ under $C$ is moved recursively. If more than $\Omega(\log n)$ steps are taken, the filter is rebuilt. Standard Cuckoo Hashing analysis shows that for any insert on a Cuckooing ACF with $\gamma = 1+\Omega(1)$ a rebuild occurs with probability $O(1/n^2)$~\cite{PaghRodler04}.

\subsection{Cyclic ACF}
\label{sec:cyclicACF}
The \defn{cyclic ACF} of Mitzenmacher et al.~\cite{MitzenmacherPo17} is an ACF where each slot contains $s$ additional \defn{hash selector} bits. The cyclic ACF generally has $b=1$; thus, the total space used by a Cyclic ACF is $kN(f + s)$.  Usually, $s$ is a small constant.

In the Cyclic ACF, the fingerprint hash maps $U\times \{0,\ldots, 2^s - 1\}\rightarrow \{0,\ldots, 2^f - 1\}$.  In particular, the hash selector bits are used to determine the fingerprint of an element stored in a given slot.

When an element $x_i$ is initially inserted,
the insertion process continues as in Section~\ref{sec:ACFoperations}, with fingerprint $h^f(x_i,0)$.   When an empty slot $\sigma$ is found that can store $x_i$, the hash selector bits of $\sigma$ are set to $0$, and $h^f(x_i,0)$ is stored in $\sigma$.

To query an element $q$, for each location hash $h^{\ell}_{\beta}$, with $\beta\in \{0,\ldots, k-1\}$, the filter looks at the slot $h^{\ell}_{\beta}$ of table $\beta$.  The $s$ hash selector bits stored in the slot contain a value $0 \leq \alpha \leq 2^s - 1$.  The filter compares $h^f(q,\alpha)$ with the fingerprint stored in the slot; the filter returns \present{} if they are equal.  Otherwise the filter increments  $\beta$ and repeats.  If no collisions are found for all $0\leq\beta\leq k-1$, the filter returns \absent{}.

If a query $q$ is a false positive, the Cyclic ACF fixes the query as follows.  Let $x_i$ be the element that collides with $q$, let $\sigma$ be the slot storing $x_i$, and let $\alpha$ be the value of the $s$ hash selector bits stored in $\sigma$.  Then the filter sets the hash selector bits of $\sigma$ to store value $\alpha + 1$, and stores $h^f(x_i,\alpha + 1)$ in $\sigma$.  If multiple $x\in S$ collide with $q$, this procedure is repeated for each such $x$.

\subsection{Swapping ACF}

The idea of the Swapping ACF (of Mitzenmacher et al.~\cite{MitzenmacherPo17}) is to have elements hash to bins with $b > 1$ slots, and to have the fingerprint of an item depend on the slot it is stored in.  In this way, elements can be (potentially) fixed by moving them to a different slot.

Inserts proceed as described in Section~\ref{sec:ACFoperations}.  However, 
in the Swapping ACF, the fingerprint hash maps $U\times \{0,\ldots, b-1\}\rightarrow \{0,\ldots, 2^f - 1\}$.  During an insert, the slot storing an element must be determined before its fingerprint can be calculated.

If a query $q$ is a false positive under configuration $C$, the filter can fix the query as follows.  Let $x_i$ be the element that collides with $q$ and let $b(x_i) = h^\ell_{\beta_i}(x_i)$ be the bin currently storing $x_i$.  Let $\sigma_i \in \{0,\ldots, b-1\}$ be the index of the slot in $b(x_i)$ currently storing $x_i$; thus $x_i$ is stored in slot $h^\ell_{\beta_{i}}(x_i)\cdot b + \sigma_x$.

The filter picks a slot index $\sigma'\in \{0,\ldots,\sigma_i-1,\sigma_i+1,\ldots b-1\}$, selected at random from the slots in $b(x_i)$, excluding the slot currently storing $x_i$. Let $x_j$ be the element currently stored in that slot if it exists.  The filter then swaps the elements: it stores fingerprint $h^f(x_i,\sigma')$ in slot $h^\ell_{\beta_i}(x_i)\cdot b + \sigma'$, and fingerprint $h^f(x_j,\sigma_i)$ in slot $h^\ell_{\beta_i}(x_i)\cdot b + \sigma_i$ (if $x_j$ does not exist, $\sigma_i$ becomes unoccupied).

\section{Bounding the False Positive Rate by the Number of Distinct Queries}
\label{sec:upper}

In this section we show that the Cuckooing ACF is \ouradaptive{}: it achieves strong performance against skewed datasets, where the queries are taken from a relatively small set of elements.

Our analysis focuses on a Cuckooing ACF with $k=2$ hash tables, $b=1$ slots per bin, and $N = n$ slots per hash table\footnote{That is to say, $\gamma = 2$.} (corresponding to the classic Cuckoo Hashing analysis).  
The experiments in Section~\ref{sec:experiments} indicate that our analysis likely extends to broader parameter ranges.  
However, formally completing the analysis for all parameters would require significant new structural insights in our proofs (e.g. Lemma~\ref{lem:initialfps}); we leave this to future work.

\begin{theorem}
  \label{thm:distinct}
  Consider a sequence of at most $n$ queries $q_1,\ldots q_n$ to a Cuckooing ACF $\mathcal F$ with $k=2$ hash tables, $N=n$ slots per table, and fingerprints of length $f = \log 1/\epsilon$ bits.  Let $Q = \bigcup_{i=1}^n \{q_i\}$.
  Then the expected number of false positives incurred by $\mathcal F$ while querying $q_1,\ldots, q_n$ is $\epsilon|Q| + O(\epsilon^2 |Q| + \log^4 n)$.  
\end{theorem}
Thus, for any sequence of $n$ queries with a support of size $|Q| = \omega(\log^4 n/\epsilon)$, the Cuckooing ACF is \ouradaptive{}.  

In contrast, for a worst-case input sequence, the Cyclic ACF and the Swapping ACF do not perform much better than a Cuckoo Filter.  Taking the Cyclic ACF as an example, consider a sequence of $n$ queries, each chosen uniformly at random from a randomly-selected set of size $|Q| = 1/\epsilon^{2^s}$.
Each of these queries collides with some $x\in S$ under \emph{every} choice of hash selector bits with probability $\Omega(\epsilon^{2^s})$.
Thus, over $n$ queries, the Cyclic ACF incurs
$\Omega(n)$ false positives for constant $\epsilon$ and $s$, compared to $O(\log^4 n)$ false positives for the Cuckooing ACF via Theorem~\ref{thm:distinct}.  
See the proof of Theorem~\ref{thm:cycliclower} for a more detailed explanation. 

\subsection{Proof of Theorem~\ref{thm:distinct}}

Without loss of generality we assume that each false positive query only collides in one of the hash tables.
Since $k=2$, fixing a query that collides in both hash tables can be simulated by executing the fixing function for each hash table separately.  

To simplify notation, we define $B(i,C) = h^{\ell}_{C[i]}(x_i)$ to be the bin storing $x_i$ under configuration $C$, and $B'(i,C) = h^{\ell}_{1-C[i]}(x_i)$ to be the alternate bin that can store $x_i$.  Because $b=1$ for this analysis, we can refer to ``slot'' $B(i,C)$ and ``bin'' $B(i,C)$ equivalently.

Let $C_0$ be the configuration of $\mathcal{F}$ before the first query $q_1$, and for $1\leq i \leq n$ let $C_i$ be the configuration after query $q_i$.
For each $1\leq i\leq n$, if $q_i$ is a false positive under $C_{i-1}$, let $k_i$ be the number of elements moved when fixing query $q_i$; otherwise let $k_i=0$.  We denote the sequence of elements moved when fixing $q_i$ as $x_{i_1}, x_{i_2}, \ldots, x_{i_{k_i}}$. Thus, $q_i$ collides with $x_{i_1}$ under $C_{i-1}$.  
We call the sequence of slots affected by these movements 
$B(i_1,C_{i-1}),B(i_2,C_{i-1}),\ldots B(i_{k_i},C_{i-1}), B'(i_{k_i},C_{i-1})$ 
the \defn{path} on $C_{i-1}$ of $q_i$.\footnote{The final term $B'(i_{i_{k_i}},C_{i-1})$ denotes that the last element is swapped to a new position, and does not ``cuckoo'' any further elements.}

We say that $q_i$ \defn{loops} if one of the moved elements repeats; i.e. there exist $1\leq \ell_1 < \ell_2 \leq k_i$ such that ${i_{\ell_1}} = {i_{\ell_2}}$.  Interestingly, classic Cuckoo Hashing analysis generally only needs to bound the number of queries that loop twice, as only twice-looping queries force a rebuild.  However, even a query that loops once cannot be fixed in a Cuckooing ACF, so we must bound how frequently this happens in our analysis.

Let the \defn{initial false positives} be the queries in $Q$ that are false positives for $\mathcal{F}$ in configuration $C_0$.
\iffull
For a fixed $C$, $\mathcal{F}$, and $Q$, 
we denote the set of initial false positives as 
\[
  F_0(C,\mathcal{F},Q) = \{q\in Q ~|~ q \text{ is a false positive for $\mathcal{F}$ under } C_0\}.
\]  
Fixing $C$, $\mathcal{F}$, and $Q$, we refer to $F_0 := F_0(C,\mathcal{F}, Q)$.
\fi

We start with a structural lemma: the elements moved when fixing any query consist of a (possibly empty) sequence of elements stored in the slot they occupied in $C_0$, followed by a (possibly empty) sequence of elements not stored in the slot they occupied in $C_0$.

\begin{lemma}
  \label{lem:initialfps}
  If a query $q_i$ on a configuration $C_{i-1}$ moves an element $x_{i_\ell}$ satisfying $C_{i-1}[i_\ell] \neq C_0[i_\ell]$, and $q_i$ does not loop, then all $j$ with $\ell\leq j\leq k_i$ satisfy $C_{i-1}[i_j] \neq C_0[i_j]$. 
\end{lemma}
\begin{proof}
  This proof is by induction on $j$; the base case $j=\ell$ is satisfied by assumption.

  Assume by induction that $C_{i-1}[i_{j-1}] \neq C_0[i_{j-1}]$ for some $j > \ell$.  Since $q_i$ does not loop, when $x_{i_{j-1}}$ is moved, it cannot have been moved previously while fixing $q_i$, and thus must be stored in slot 
	$B(i_{j-1},C_{i-1})$.
	Then  after $x_{i_{j-1}}$ is moved it must be stored in slot 
	$B(i_{j-1},C_0)$;
	this must be equal to the slot storing $x_{i_j}$. Because $q_i$ does not loop, $x_{i_j}$ must be stored where it was when the fixing began; i.e. in 
	$B(i_j,C_{i-1})$.
	Thus 
	$B(i_j,C_{i-1}) = B(i_{j-1},C_0)$,
	so $C_0[i_j] \neq C_{i-1}[i_j]$, as otherwise $x_{i_j}$ and $x_{i_{j-1}}$ would be stored in the same slot in $C_0$.
\end{proof}

Lemma~\ref{lem:initialfps} immediately gives structure to the problem in two key ways.  First, it limits how queries can break one another: if $q$ is a false positive, but is not an initial false positive, then there must be some initial false positive $q_i$ that caused $q$ to become a false positive.  
We do not need to worry about non-initial false positives causing other, new false positives.
Second, it ties the behavior of all elements to how they behave on the initial configuration $C_0$.  This means that we can make statements about how queries interact using $C_0$; we do not need to reset our analysis every time the filter configuration changes.

\iffull

Our analysis depends heavily on how queries behave on configuration $C_0$, so
we introduce notation to help discuss queries on $C_0$.  
Let $k^0_i$ be the number of elements moved when querying $q_i$ on configuration $C_0$.   Let $x_{i_1'},x_{i_2'}, \ldots x_{i_{k^0_i}'}$ be the set of elements moved when fixing $q_i$, if $q_i$ were queried on configuration $C_0$.

We now define the notion of {costly queries}.  In short, 
costly queries are either difficult to fix, or break other queries.
However, costly queries will prove to be rare, and will not substantially increase the total number of false positives.

A query $q_i$ is a \defn{costly query} if $q_i$ is a false positive on $C_{i-1}$ and $q_i$ meets one of the following criteria.
We refer to these as Criteria 1-4 in our analysis.

\begin{enumerate}
  \item \label{criterion:pathscollide} 
    If $q_i$ is an initial false positive, then 
		there exists a $q_j$ such that the path of $q_j$ when queried on $C_0$ collides with the path of $q_i$ when queried on $C_0$;
		if $q_i$ is not an initial false positive, then there exists a $q_j$ such that 
		the path of $q_j$ when queried on $C_0$ is hashed to by $q_i$.
    In other words, if $q_i\in F_0$ there exists a $q_j\in F_0\setminus\{q_i\}$, and a pair of indices $\ell_j, \ell_i$ with $1\leq \ell_j\leq k^0_j$ and $1\leq \ell_i\leq k^0_i$ and a hash index $\beta\in\{0,1\}$ where $h^\ell_{\beta}(x_{j_{\ell_j}'}) = h^\ell_{\beta}(x_{i_{\ell_i}'})$; if $q_i\notin F_0$ then there exists a $q_j\in F_0\setminus\{q_i\}$ and an index $\ell_j$ with $h^{\ell}_{C_0[\ell_j]}(q_i) = h^{\ell}_{C_0[\ell_j]}(x_{j'_{\ell_j}} )$.
  \item \label{criterion:loops} $q_i$ loops when queried on $C_{i-1}$.
  \item \label{criterion:loopcollision} 
                There exists a false positive $q_j\in Q$ that loops on $C_{j-1}$ and whose path contains a slot that $q_i$ hashes to.
    In other words, there exists a $q_j\in Q$ that loops on $C_{j-1}$, an integer $1 \leq \ell\leq k_j$, and a hash index $\beta\in\{0,1\}$ such that 
		$B(j_\ell,C_{j-1}) = h^{\ell}_{\beta}(q_i)$ or $B'(j_\ell,C_{j-1}) = h^{\ell}_{\beta}(q_i)$.
  \item \label{criterion:selfcollision} 
    $q_i$ hashes to a slot along its own path (excluding the slot storing the first moved element): thus, 
    there exists a $1\leq \ell\leq k_i$  such that 
		$B'(i_\ell,C_{i-1}) = h^{\ell}_{1-C_{i-1}[i_\ell]}(q_i)$.
\end{enumerate}

To begin, we show that all all false positives queries are either an initial false positive, or a costly query.
\begin{lemma}
  \label{lem:twokindsfps}
  Let $q_i$ be a false positive on $C_{i-1}$ for filter $\mathcal{F}$.  Then $q_i\in F_0$ or $q_i$ is costly. 
\end{lemma}
\begin{proof}
  Assume that $q_i$ is a false positive with $q_i\notin F_0$; we show that $q_i$ is costly.  Let $1\leq j < i$ be the smallest $j$ such that $q_i$ is a false positive under $C_{j}$ ($j$ must exist because $q_i$ is a false positive under $C_{i-1}$, but $q_i$ is not an initial false positive).  If $q_j$ loops, then $q_i$ is costly by Criterion~\ref{criterion:loopcollision}.  

  Now, assume $q_j$ does not loop.  There must be some $x_{j_\ell}$ moved by $q_j$ (when queried on $C_{j-1}$) such that $q_i$ collides with $x_{j_\ell}$ in $C_j$.  Since $q_i$ is not an initial false positive and $q_j$ does not loop, we must have that $q_i$ collides with $x_{j_\ell}$ when $x_{j_\ell}$ is stored in slot 
	$B'(j_\ell, C_{j-1}) = B'(j_\ell,C_0)$.  
	However, this means that $C_{j-1}[j_\ell] = C_0[j_\ell]$, and therefore $C_{j-1}[j_1]=C_0[j_1]$ as otherwise we reach a contradiction with Lemma~\ref{lem:initialfps}.
	This means that $q_j$ is an initial false positive, and thus $q_i$ is costly by Criterion~\ref{criterion:pathscollide}.
\end{proof}

We bound the cost of all queries using a potential function argument. The potential of a configuration is the number of elements $x\in S$ stored in their original position that collide with an initial false positive.
  Specifically, for all $0\leq t \leq n$, define 
  \[
    \Phi(t) = \bigg|\{x_i\in S, q\in F_0 ~|~ x_i \text{ collides}
        \text{ with } q \text{ under } C_{t} \text{ and } C_0[i] = C_t[i]\}\bigg|.
    \]

  Let the amortized cost of query $q_i$ be $1+ \Phi(i) - \Phi(i-1)$ if $q_i$ is a false positive, and $\Phi(i) - \Phi(i-1)$ if $q_i$ is not a false positive.  Summing (and taking $\Phi(n) \geq 0$), the expected number of false positives incurred during the queries on $q_1,\ldots, q_{n}$ is at most the expected number of amortized false positives plus $\E[\Phi(0)]$. 
  The proof proceeds in three parts: first, bounding $\E[\Phi(0)]$, then bounding the cost of queries that are not costly, and finally bounding the cost of costly queries.

	The expected potential before any queries follows immediately from linearity of expectation and $f= \log (1/\epsilon)$.
  \begin{lemma}
    \label{lem:potential}
    $ \E[\Phi(0)] \leq \epsilon|Q|$
  \end{lemma}
  \iffull
  \begin{proof}
  For a given $x_i\in S$ and $q\in Q$, the probability that $h_{C_0[i]}(x_i) = h_{C_0[i]}(q)$ is at most $\epsilon/n$.  
  Summing over the $|Q|n$ pairs obtains $\E[\Phi(0)]\leq \epsilon|Q|$.
  \end{proof}
  \fi

  Now we show that non-costly queries have no amortized cost.  This is immediate for any $q_i$ that is not a false positive.
	For a non-costly false positive $q_i$, we show that this cost is offset by a decrease in potential from moving $x_{i_1}$.
  \begin{lemma}
    \label{lem:noncostlycost}
    If a query $q_i$ is not costly, then the amortized cost of $q_i$ is at most $0$.
  \end{lemma} 

  \iffull
  \begin{proof}
  If $q_i$ is not a false positive, then no elements are moved, and $\Phi(i) = \Phi(i-1)$. Thus the amortized cost of a query that is not a false positive is $0$.  

  If $q_i$ is a false positive but is not costly, then $q_i$ is an initial false positive by Lemma~\ref{lem:twokindsfps}; furthermore $q_i$ does not loop by Criterion~\ref{criterion:loops}.  
  Since $q_i$ is an initial false positive, moving $x_{i_1}$ decreases $\Phi$ by 1 ($q_i$ does not collide with $x_{i_1}$ in its new position by Criterion~\ref{criterion:selfcollision} and $x_{i_1}$ is not moved again while fixing $q_i$ because $q_i$ does not loop by Criterion~\ref{criterion:loops}).

  We now show that this is the only change to $\Phi$.     Assume the contrary: there is a $q_j\in F_0$ and an $\ell \leq k_i$ such that $x_{i_\ell}$ collides with $q_j$ under $C_i$, and $C_i[i_\ell] = C_0[i_\ell]$.  By Lemma~\ref{lem:initialfps}, since $q_i$ does not loop we can partition the elements moved by $q_i$ into at most two sequences of elements: the possibly-empty sequence $x_{i_1},\ldots x_{i_\lambda}$,  and the sequence $x_{i_{\lambda + 1}},\ldots x_{i_{k_i}}$ where for all $\kappa\in \{1,\ldots,\lambda\}$, $C_i[i_\kappa] \neq C_0[i_\kappa]$, and for all $\kappa\in\{\lambda+1,\ldots,k_i\}$, $C_i[i_\kappa] =  C_0[i_\kappa]$.  Thus, $\ell > \lambda$.  
  Recall that ${j_1'},\ldots {j_{k^0_j}'}$ is the sequence of elements moved when querying $q_j$ on $C_0$.
  The sequence ${i_\ell}, {i_{\ell-1}},\ldots, {i_{\lambda + 1}}$ (notice that $\ell,\ell-1 \ldots, \lambda+1$ is in decreasing order) must be a subsequence of  ${j_1'},\ldots {j_{k^0_j}'}$. 
  Similarly, if ${i_1,\ldots, i_\lambda}$ exists it must be a subsequence of $i_1',\ldots, i_\lambda'$.
  If $q_i = q_j$, then $q_i$ must be costly by Criterion~\ref{criterion:selfcollision}.  If $q_i\neq q_j$, then $q_i$ must be costly by Criterion~\ref{criterion:pathscollide}, since if $x_{i_\lambda}$ exists then $h^\ell_{C_0[i_\lambda]}(x_{i_\lambda}) = h^\ell_{1-C_0[i_{\lambda+1}]}(x_{i_{\lambda+1}})$, and if $x_{i_\lambda}$ does not exist then $h^{\ell}_{1-C_0[i_{\lambda+1}]}(q_i) = h^\ell_{1-C_0[i_{\lambda+1}]}(x_{i_{\lambda+1}})$.  Thus $\Phi(i) - \Phi(i-1)\leq -1$, and the amortized cost of $q_i$ is at most $0$.
  \end{proof}
  \fi

  We begin analyzing costly queries by bounding their cost in terms of their path length.  Intuitively, this bound reflects that each of the $k_i$ elements moved when fixing $q_i$ collides with an element of $Q$ with probability $\epsilon/n$.

  \begin{lemma}
    \label{lem:costlycost}
    If $q_i$ is costly, then its expected amortized cost is at most $4 + 2\epsilon |Q|k_i/n$.
  \end{lemma}
  
\iffull
  \begin{proof}
  We upper bound the difference in $\Phi$ using a very pessimistic case: if some moved element $x_{i_\ell}$ collides with some $q'\in F_0$ under some hash $h_{\beta}$, then we assume that $x_{i_\ell}$ collides with $q'$ under $C_i$ but not $C_{i-1}$.  This leads to the bound
  \[
    \E\left[\Phi(i) - \Phi(i-1)\right] \leq 
      \E\bigg[\sum_{\ell = 1}^{k_i} |\{q\in F_0 ~|~ h_{\beta}(q) = h_{\beta}(x_{i_\ell}) \text{ for some } \beta\in\{0,1\} \}| ~\bigg|~q_i\text{ is costly}\bigg].
  \]

  If the path of $q_i$ collides with the path of $q_j$ (under Criterion~\ref{criterion:loopcollision}), then $x_{j_1}$ may collide with $q_j$ under $C_i$ and $C_0$; the same holds if the path of $q_i$ collides with the path of $q_j$ when queried on $C_0$ (under Criterion~\ref{criterion:pathscollide}).  If $q_i$ hashes to a slot along its path, it may collide with the moved element under both $C_0$ and $C_i$.  
If $q_i$ loops (under Criterion~\ref{criterion:loops}), then $q_i$ collides with $x_{i_1}$ under $C_{i-1}$ and may collide with $x_{i_1}$ under $C_i$ and $C_0$; however this does not increase $\Phi$.  
  We sum to obtain a (loose) upper bound of $3$ collisions from the cases listed in this paragraph.

  For any other pair $(q',\lambda)$ where $q'\in F_0$ and $\lambda\in \{1,\ldots k_i\}$, $q'$ and $x_{i_\lambda}$ collide under $\beta$ with probability $\epsilon/n$ for each hash index $\beta\in \{0,1\}$.  
  Summing via union bound over all $q',\lambda,\beta$ obtains 
  \[
      \E\bigg[\sum_{\ell = 1}^{k_i} |\{q\in F_0 ~|~ h_{\beta}(q) = h_{\beta}(x_{i_\ell}) \text{ for }
      \text{some } \beta\in\{0,1\} \}| ~\bigg|~q_i\text{ is costly}\bigg] \leq 3 + (2k_i|Q|)(\epsilon/n)
  \]
Adding in the cost of the false positive $q_i$, the expected amortized cost of $q_i$ is at most $4 + 2\epsilon|Q|k_i/n$.  
\end{proof}
\fi

The following lemma is dedicating to bounding the increase in potential from costly queries that satisfy Criterion~\ref{criterion:pathscollide}.  
These are by far the most common and expensive costly queries, so our analysis of this case needs to be tight up to (essentially) constants to obtain our desired bounds. 
That said, we did not focus on minimizing the constants themselves.

\begin{lemma}
  \label{lem:pathscollidebound}
  Assume that all initial false positives $q_i\in F_0$ satisfy $k_i^0 = O(\log n)$.  Then the expected number of queries satisfying Criterion~\ref{criterion:pathscollide} is at most $8|F_0|^2/n ~+~ 2\epsilon |Q| |F_0|/n$.  Furthermore, 
  \[
	  \E[k_i ~|~ \text{$q_i$ satisfies Criterion~\ref{criterion:pathscollide}}] \leq 96.
  \]
\end{lemma}
\iffull
\begin{proof}
  We begin by bounding the number of queries $q_i\in F_0$ that satisfy Criterion~\ref{criterion:pathscollide}.  
  Taking a union bound, for any $q_j\in F_0$, the probability that an element moved when querying $q_i$ shares a slot with an element moved when querying $q_j$ is $k_i^0 k_j^0/n$.  
  So long as $q_i$ and $q_j$ are moving distinct elements when being fixed under $C_0$, the hashes of each element are independent, and therefore the number of elements each moves is independent.%
\footnote{Strictly speaking, these events are not independent.  After all, if $k^0_i = n$, then $k^0_j$ is very likely to be $\Omega(n)$.  However, since we assume $k^0_i = O(\log n)$ and $k^0_j = O(\log n)$, only $O(\log n)$ hash table slots are touched by either path; this affects the expectation of each by at most $O((\log n)/n)$.}
  If $q_i$ and $q_j$ ever move the same element, they satisfy Criterion~\ref{criterion:pathscollide} and we are done.
  Thus, we can upper bound the expected number of pairs $q_i,q_j\in F_0$ satisfying Criterion~\ref{criterion:pathscollide} by $\sum_{q_i,q_j\in F_0} \E[k_i^0 k_j^0]/n$.
  
  From~\cite[Section 3.1.2]{PaghRodler04} we have (in the notation of this paper, and recalling that the filter may access fully random (i.e. $(1,n)$-universal) hash functions, and recalling that we assume that $q_i$ and $q_j$ collide with at least one element), $\Pr[k^0_i = \hat{k}] \leq 2(1/\gamma)^{-\hat{k} - 1} = 2^{2-\hat{k}}$ for $\hat{k} \geq 1$.  As above, we can assume that $k^0_i$ and $k^0_j$ are independent for the purpose of our upper bound. Summing,
  \begin{align*}
    \E[k_i^0 k_j^0] &\leq \sum_{k_i^0 = 1}^{\infty} \sum_{k_j^0 = 1}^{\infty} (k_i^0) (k_j^0) 2^{2-k_i^0} 2^{2- k_j^0}\\
    &= 16\sum_{k_i^0 = 1}^{\infty}(k_i^0) 2^{-k_i^0}\left(  \sum_{k_j^0 = 1}^{\infty}  (k_j^0)  2^{ - k_j^0}\right)\\
    &\leq 32  \sum_{k_i^0 = 1}^{\infty}(k_i^0) 2^{-k_i^0} \leq 64
  \end{align*}
  Thus, the expected number of pairs that satisfy Criterion~\ref{criterion:pathscollide} is at most $64|F_0|^2/n$ for $n \geq 64$.

  We now consider the case when $q_i\notin F_0$.  This means that $q_i$ must hash to the same slot as an element moved by some $q_j$ when $q_j$ is being fixed on $C_0$; for a given element this occurs with probability $\epsilon/n$.  This means that the probability that $q_i$ satisfies Criterion~\ref{criterion:pathscollide} is at most $\sum_{q_j\in F_0} \epsilon k_j^0/n$.  Summing obtains $2\epsilon |Q| |F_0|/n$ expected false positives over all $q_i\in Q$.

  Finally, we bound $\E[k_i ~|~ \text{$q_i$ satisfies Criterion~\ref{criterion:pathscollide}}]$.  By definition, 
  \begin{align*}
	  \E[k_i ~|~ \text{$q_i$ satisfies Criterion~\ref{criterion:pathscollide}}] &= \sum_{k' = 1}^{\infty} k' \Pr[k_i = k' ~|~ \text{$q_i$ satisfies Criterion~\ref{criterion:pathscollide}}]\\
																				&= \frac{\Pr[k_i = k' \text{ and $q_i$ satisfies Criterion~\ref{criterion:pathscollide}}] } {\Pr[\text{$q_i$ satisfies Criterion~\ref{criterion:pathscollide}}] }
\end{align*}
  As above, 
    \begin{align*}
      \Pr[k_i = k' \text{ and $q_i$ satisfies Criterion~\ref{criterion:pathscollide}}]  &\leq k' 2^{2-k'} \sum_{q_j\in F_0} \sum_{k_j^0 = 0}^{\infty} k_j^0 2^{2-k_j^0}/n \\
      &\leq 4 k' 2^{-k'} |F_0| \sum_{k_j^0 = 0}^{\infty} k_j^0 2^{-k_j^0}/n \\
      &\leq 8  |F_0| k' 2^{-k'}/n
      \end{align*}
  A lower bound for $\Pr[q_i \text{ satisfies Criterion~\ref{criterion:pathscollide}}]$ is the event that 
  $h^{\ell}_{1-C[j_1]}(q_i) = h^\ell_{1-C[j_1]}(x_{j_1})$.  
  This occurs with probability $ 1 - \left( 1- 1/n\right)^{|F_0|} \geq 1 - 1/2^{|F_0|/n} $.
  Thus,
  \begin{align*}
    \E[k_i ~|~ \text{$q_i$ satisfies Criterion~\ref{criterion:pathscollide}}] &= \sum_{k' = 1}^{\infty} k' \Pr[k_i = k' ~|~ \text{$q_i$ satisfies Criterion~\ref{criterion:pathscollide}}]\\
    &\leq \sum_{k' = 1}^{\infty} k' \frac{8  |F_0| k' 2^{-k'}/n} {1 - 1/2^{|F_0|/n}}\\
    &\leq 8 \frac{|F_0|/n} {1 - 1/2^{|F_0|/n}} \sum_{k' = 1}^{\infty} k'^2 2^{-k'} \\
    &\leq 16  \sum_{k' = 1}^{\infty} k'^2 2^{-k'} \leq 96\qedhere
  \end{align*} 
\end{proof}
\fi

Now we bound the total cost of all costly queries.  The additive polylog term is due to applying tail bounds to rare but problematic events like looping queries---it is possible that a sufficiently involved analysis of these events would reduce or remove this term.

\iffull
\begin{lemma}
  \label{lem:boundcost}
  For any sequence of queries $q_1,\ldots, q_{n}$ 
  \[
    \E\left[\sum_{\substack{q_i\text{ is costly}\\1\leq i\leq n}} (4 + 2\epsilon|Q|k_i/n)\right] \leq O(\epsilon^2|Q|^2/n + \log^4 n).
  \]
\end{lemma}
\fi
\iffull
\begin{proof}
  From classic Cuckoo Hash analysis, $k_i = O(\log n)$ for all $i$ with probability at least $1-1/n$ (see i.e.~\cite{PaghRodler04}).  We upper bound the number of costly $q_i$ by summing the expected number of $q_i$ meeting each criterion.  Assume that $k_i = O(\log n)$ and $k^0_i = O(\log n)$ for all $i$ (if not, then at most all $n$ queries are false positives; since this edge case happens with probability $\leq 2/n$ this adds $2$ to the expected cost).  
  
  \textbf{Criterion~\ref{criterion:pathscollide}:} 
  From Lemma~\ref{lem:pathscollidebound}, the expected number of queries satisfying Criterion~\ref{criterion:pathscollide} is at most $8|F_0|^2/n + 2\epsilon |Q| |F_0|/n$.  Substituting $\E[k_i ~|~ \text{ $q_i$ satisfies Criterion~\ref{criterion:pathscollide}}] \leq 96$ in Lemma~\ref{lem:costlycost}, each of these queries costs at most $(3 + 192\epsilon|Q|/n) = O(1)$; thus the total cost is $O(|F_0|^2/n + \epsilon |Q| |F_0|/n)$.  

  To bound $E[|F_0|^2]$, notice that this expression represents the square of the sum of $|Q|$ independent Bernoulli trials, each of which succeeds with probability $\epsilon$.  Thus, the variance of $|F_0|$ is $\epsilon|Q|(1-\epsilon)$; substituting $\E[|F_0|] = \epsilon|Q|$ into $\text{Var}(X) = \E[X^2] - (\E[X])^2$ obtains $\E[|F_0|^2] \leq \epsilon^2|Q|^2 + \epsilon(1-\epsilon)|Q|$. By linearity of expectation, $\E[ |Q| |F_0|/n] \leq \epsilon |Q|^2/n$.  Summing obtains a total cost of $O(\epsilon^2 |Q|^2/n)$.

  \textbf{Criterion~\ref{criterion:loops}:} Query $q_i$ loops if and only if there exists an element $x_{i_\ell}$ (with $1\leq\ell\leq k_i$) such that the second hash location of $x_{i_\ell}$ collides with another one of the $k_i$ slots containing elements moved by $q_i$.  For a given slot $\sigma$, the probability that $h^\ell_{C_i[i_\ell]}(x_{i_\ell}) = \sigma$ is $1/n$. Taking a union bound, $q_i$ loops with probability at most $k_i^2/n$.  Summing over all $2n$ queries, the expected number of queries that loop is $O(\log^2 n)$.   Since $k_i = O(\log n)$, Lemma~\ref{lem:costlycost} gives that each costs $O(\log n)$, giving a total cost of $O(\log^3 n)$.

  \textbf{Criterion~\ref{criterion:loopcollision}:}  If there are $L$ queries that loop, since each looping query moves $O(\log n)$ elements, the probability that $q_i$ collides with any element moved by a looping query is $O(L\log n/n)$.  By linearity of expectation, since $\E[L] = O((\log n)^2)$ (shown in the proof for Criterion~\ref{criterion:loops}), summing over all $q_i$, at most $O(\log^3 n)$ queries meet Criterion~\ref{criterion:loopcollision} in expectation.  Using Lemma~\ref{lem:costlycost} obtains a total cost of $O(\log^4 n)$.

  \textbf{Criterion~\ref{criterion:selfcollision}:} Each element $x_{i_\ell}$ moved when querying $q_i$, from $1\leq\ell\leq k_i$, hashes to the same slot as $q_i$ with probability $1/n$; thus $q_i$ collides with an element after it is moved with probability $O(\log n)/n$.  Summing over all queries, there are $O(\log n)$ queries that are costly in expectation by Criterion~\ref{criterion:selfcollision}; using Lemma~\ref{lem:costlycost} obtains a total cost of $O(\log^2 n)$.

  Summing, all costly queries have a total cost of $O(\epsilon|Q|^2/n + \log^4 n)$.
\end{proof}
\fi

Summing the starting potential cost (Lemma~\ref{lem:potential}), the total amortized cost of all queries that are not costly (Lemma~\ref{lem:noncostlycost}), and the total amortized cost of all costly queries (Lemma~\ref{lem:boundcost}), we obtain Theorem~\ref{thm:distinct}.
\fi

\section{Experiments}
\label{sec:experiments}

In this section, we examine how the Cuckooing ACF performs on network trace datasets.  There are two main takeaways from this section.  First, the design of the Cuckooing ACF results in better practical performance than previous adaptive filters on network trace datasets.  Second, the analysis of Section~\ref{sec:upper} extends to practice: an adaptive cuckoo filter with practical parameter settings (including very high load factor) still achieves strong performance.

\begin{figure}[ht]
	\vspace{-.2in}
	\includegraphics[width=.49\textwidth]{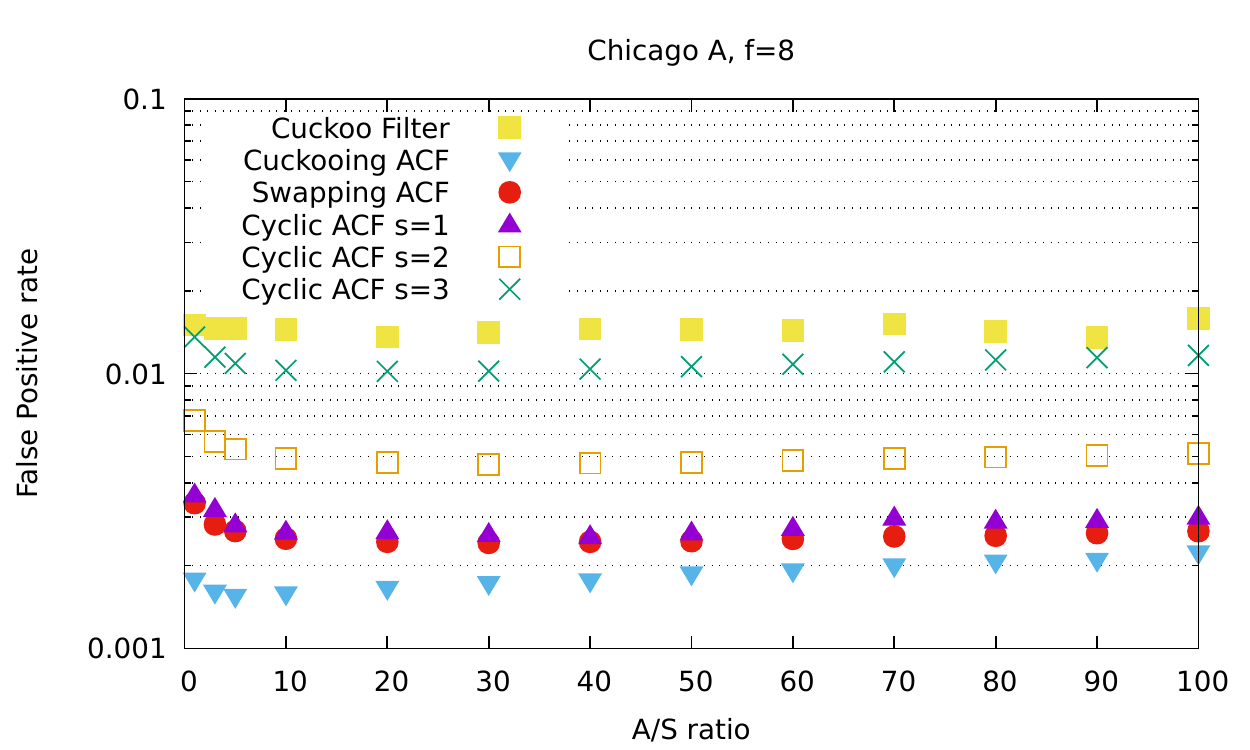}
	\iffull
	\includegraphics[width=.49\textwidth]{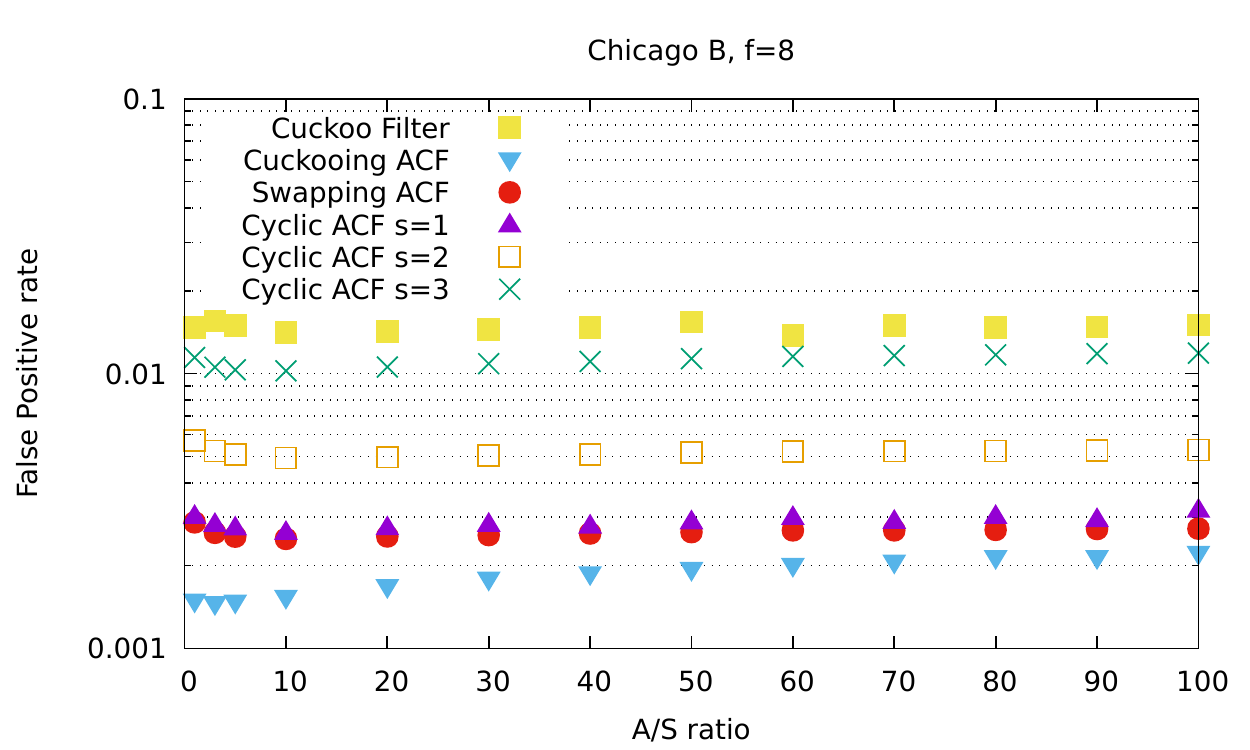}
	\includegraphics[width=.49\textwidth]{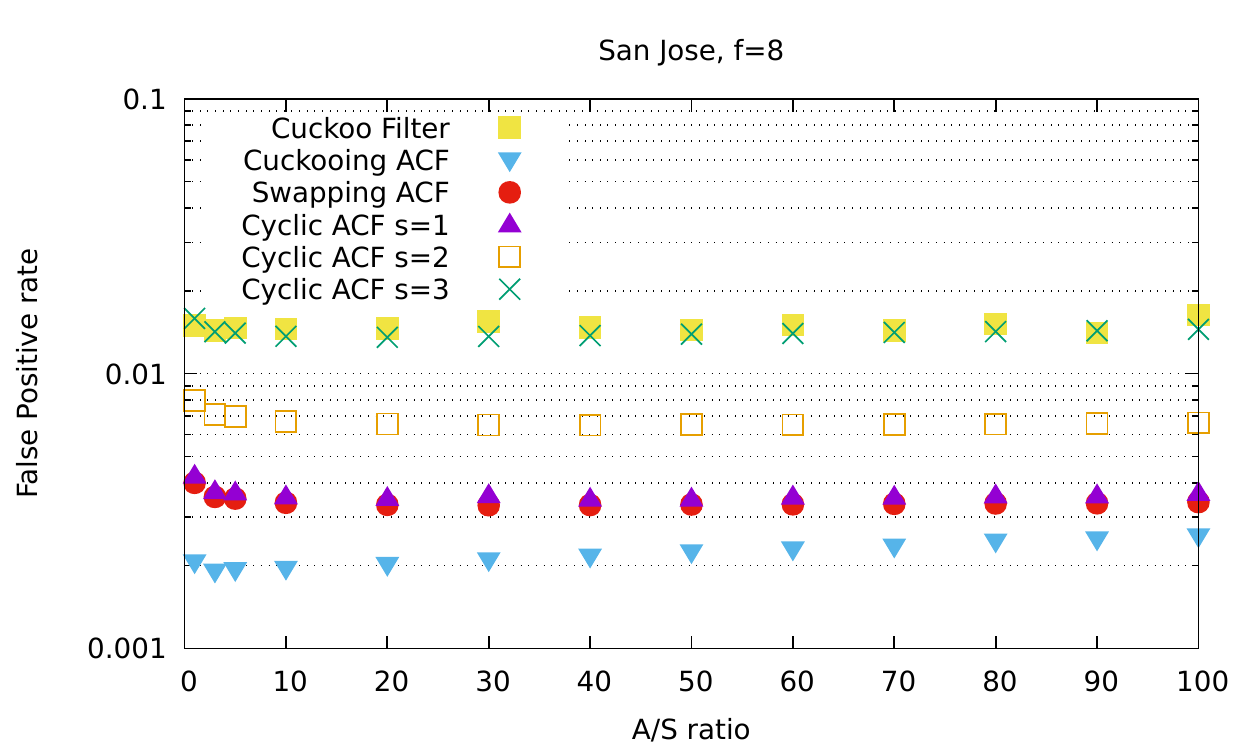}
	\fi
	\includegraphics[width=.49\textwidth]{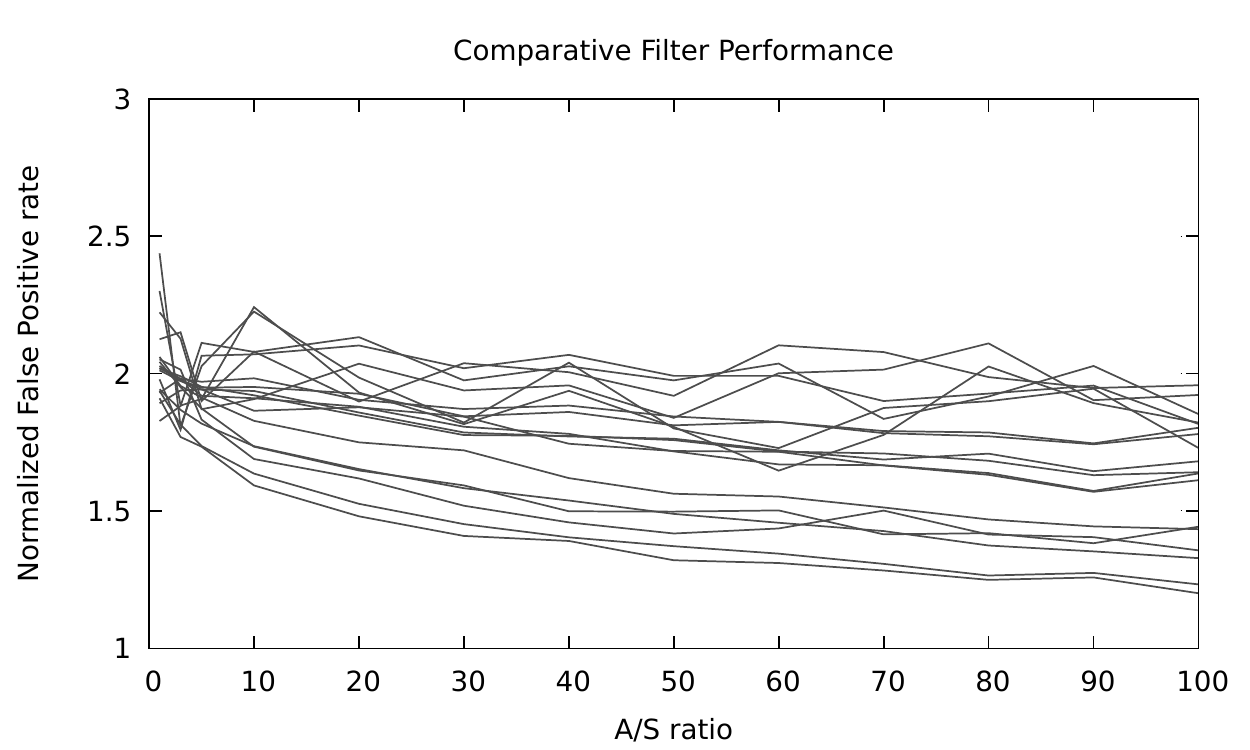}
        \caption{We examine the false positive rate of each adaptive filter, varying the ratio of the number of queries to the number of stored elements.  The bottom right figure normalizes the false positive rate by the number of false positives incurred by the Cuckooing ACF\@.  It summarizes the results for the Swapping ACF and the Cyclic ACF with $s=1$, for all three datasets, for $f=8,12,16$.} 
	\label{fig:experiments}
\end{figure}

Our experiments use three network traces from the CAIDA 2014 dataset, specifically:
\begin{itemize}[noitemsep, topsep=1pt]
	\item equinix-chicago.dirA.20140619 (which we call ``Chicago A'') 
	\item equinix-chicago.dirB.20140619-432600 (``Chicago B''), and 
	\item equinix-sanjose.dirA.20140320-130400 (``San Jose'') 
\end{itemize}
following the experiments of Mitzenmacher et al. in~\cite{MitzenmacherPo17}.   
In each test, the elements stored in the filter and the query elements both come from the network trace.  Let $A$ be the set of query elements.
We perform tests for different $|A|/|S|$ ratios; specifically $|A|/|S| = \{1,3,5,10,20,30,40,50,60,70,80,90,100\}$.

We begin by setting $n = |S|$ using the prescribed $|A|/|S|$ ratio and the total number of unique flows in the trace---in particular, $n = (\text{\# unique flows})/(1 + |A|/|S|)$.  The first $n$ unique flows seen in the trace are inserted into each filter.  The remaining flows in the trace (those not in $S$) are used as queries.  

We consider six data structures in our experiments: 
\begin{itemize}
\item a classic Cuckoo Filter, with $k=4$ hash tables and $b=1$ slots per bucket;
\item the Cuckooing ACF, with $k=4$ hash tables and $b=1$ slots per bucket;
\item three implementations of the Cyclic ACF described in~\cite{MitzenmacherPo17}, with $s=1$, $s=2$, and $s=3$ hash selector bits.  To ensure a fair comparison in space usage, each hash selector bit used is accounted for with a corresponding decrease in the number of fingerprint bits (for example, a Cuckoo Filter with fingerprints of length $8$ is compared to a Cyclic ACF with fingerprints of length $7$ and $s=1$ hash selector bit); and
\item a Swapping ACF as described in~\cite{MitzenmacherPo17}, with $b=4$ slots per bin and $k=2$ hash tables.
\end{itemize}

All filters are at 95\% occupancy in all of our experiments (i.e. $\gamma = 1/.95$).  
We give results for fingerprints of length $f=8$ bits on Chicago A, and summarize key results for Chicago B and San Jose, as well as results with $f=12$ and $f=16$ bits on all datasets.
We repeat each experiment 10 times; all results given are the average performance over these 10 trials.

\subsection{Experimental Results}

\iffull
The three plots in Figure~\ref{fig:experiments} excluding the bottom-right plot show that the Cuckooing ACF has the strongest performance of all adaptive filters on network trace datasets with fingerprints of size $f=8$.  
\fi
The left hand plot in Figure~\ref{fig:experiments} show that the Cuckooing ACF has the strongest performance of all adaptive filters on the Chicago A dataset with $f=8$.
Its performance is particularly strong for low values of $|A|/|S|$---that is to say, its performance is strong when the number of queries is small relative to $n$ (as one may expect given Theorem~\ref{thm:distinct}).

We ran further experiments, using fingerprints of size $f=8$, $f=12$, and $f=16$ on Chicago A, Chicago B, and San Jose datasets, 
achieving similar (in fact slightly better) results.  These experiments are summarized in the right hand plot of Figure~\ref{fig:experiments}; we also give full results in Appendix~\ref{sec:networkexperiments}.  The y-axis in this figure indicates the false positive rate of the given filter divided by the false positive rate of the Cuckooing ACF.  This plot only includes the two best filters: the Swapping ACF, and the Cyclic ACF with $s=1$.  We run the experiments for three fingerprint sizes $\{8,12,16\}$ on all three datasets, giving $18$ total lines in the plot.  
Note that there is some overlap with the left hand plot---one of the bottommost two lines in the plot correspond to the Swapping ACF with $f=8$ on Chicago A.  Specifically, the Cuckooing ACF does even better with larger fingerprints like $f=12$ and $f=16$ compared to $f=8$.

Overall, the Cuckooing ACF always performs at least as well as every other cuckoo filter on these datasets, frequently outperforming them by nearly a factor of $2$.  

\iffull
We believe that the simplicity of the Cuckooing ACF is, in fact, the source of this performance improvement.
The Swapping ACF uses two bins of size $4$, so each query is compared to twice as many fingerprints as in the Cuckooing ACF or Cuckoo Filter.\footnote{If the bins are decreased in size to, say, $2$, the Swapping ACF requires extremely frequent rebuilds with the desired load factor---the factor-2 decrease in efficiency is intrinsic to the Swapping ACF.} Meanwhile, the adaptivity bit of the Cyclic ACF causes each fingerprint to be one bit shorter, again doubling the false positive rate.  Thus, the simple swapping strategy means we can avoid metadata bits with $b=1$, giving the Cuckooing ACF has an immediate factor-2 advantage over its competitors.  
\fi

\section{Adersarial Adaptivity}
\label{sec:lower}

Previous work leaves a dichotomy: the Adaptive Cuckoo Filters of Mitzenmacher et al.~\cite{MitzenmacherPo17} work well in practice, whereas the ``Broom Filter'' of Bender et al.~\cite{BenderFaGo18} is effective even against an adversary that tries to ``learn'' a filter's internal state.  In Sections~\ref{sec:upper} and~\ref{sec:experiments} we showed that the Cuckooing ACF is practical while retaining theoretical bounds.  But  our theoretical bounds are not adversarial; they are based on the number of unique queries made to the filter.  Can our results be taken further---is there an ACF that adapts effectively even against an adversary?

In this section we give a general lower bound showing that 
an adversary can obtain a false positive rate of $\Omega(1)$ against any space-efficient ACF.  This result is closely tied to a key structural distinction: the Broom Filter is difficult to implement because the length of the stored fingerprint may be different for each element.  Our lower bound shows that this flexibility is, in fact, necessary in order to achieve adaptivity.

\subsection{Definition}
Bender et al.~\cite{BenderFaGo18} defined a notion of adaptivity that captures a worst-case adversary attempting to maximize the filter's false positive rate.  
We summarize this model in this subsection, and refer readers to~\cite{BenderFaGo18} for a more thorough discussion.

In the \defn{adaptivity game}, an adversary generates a sequence of queries.  After each query $q$, the adversary and filter both learn if $q$ was a false positive.  The filter may change its internal representation in response. The adversary will use whether or not $q$ was a false positive to determine the further queries.

At any time, the adversary may name a special element $\hat{q}$---the adversary is asserting that this query is likely to be a false positive.  The adversary ``wins'' if $\hat{q}$ is a false positive, and the filter ``wins'' if $\hat{q}$ is not a false positive.

The \defn{sustained false positive} rate of a  filter is the maximum probability $\epsilon$ with which the adversary can win the adaptivity game.  We call a filter $\mathcal F$ \defn{adaptive} if $\mathcal F$ can achieve a sustained false positive rate of $\epsilon$ for \emph{any} constant $\epsilon$.

\subsection{Lower Bounds}

To begin, we note that the Cyclic ACF is not adaptive. A nearly-identical proof shows that the Swapping ACF is not adaptive.

\begin{theorem}
  \label{thm:cycliclower}
	Let $\mathcal{F}$ be a Cyclic ACF with $k=O(1)$ hash tables, each with $N = \Theta(n)$ slots.  
  Then there exists an adversarial strategy, making $\Theta(2^s/\epsilon^{2^s})$ queries, which wins the adaptivity game against $\mathcal{F}$ with probability $\Omega(1)$.  Thus the sustained false positive rate of $\mathcal{F}$ is $\Omega(1)$.
\end{theorem}
\begin{proof}
  The adversary begins by selecting a query element $q_1$ at random.  The adversary queries $q_1$ $2^s$ times.  If $q_1$ is a false positive every time it is queried, the adversary sets $\hat{q} \gets q_1$; otherwise the adversary chooses a new query element $q_2$ and repeats.  This process is repeated until $O(1/\epsilon^{2^s})$ query elements have been chosen, requiring $O(2^s/\epsilon^{2^s})$ queries overall.

  We show that the adversary finds a $\hat{q}$ with probability $\Omega(1)$, and that $\hat{q}$ will be a false positive with probability $1$.

  Each time $q$ collides with an element $x_i\in S$,
  the hash selector bits associated with $x_{i}$ are incremented; thus, if $q$ does not collide with $x_i$ on the $j$th query, it will not collide on the $j'$th query for $j' > j$.  Then if $q$ is a false positive on all $2^s$ collisions, there is an $x_j\in S$ such that $h^f(q,\alpha) = h^f(x_j,\alpha)$ for all $\alpha\in\{0,\ldots, 2^s-1\}$.  
  We immediately obtain that any $\hat{q}$ found by the adversary is a false positive with probability $1$.

  We now bound the probability that $\hat{q}$ is found.
  For a given query element $q_i$ and a given element $x_j\in S$, the probability that $h^{\ell}_{\beta_j}(q_i) = h^\ell_{\beta_j}(x_j)$ is $1/n$. The probability that, for all $\alpha$, $h^f(q_i,\alpha) = h^f(x_j,\alpha)$ is $1/\epsilon^{2^s}$.  

  Thus, after the algorithm has considered $1/\epsilon^{2^s}$ query elements, the probability that there exists a query element $q'$ and an $x^*\in S$ such $x^*$ causes $q'$ to be a false positive $2^s$ times is $1 - (1 - \epsilon^{2^s}/n)^{n/\epsilon^{2^s}} = \Omega(1)$.  Thus, the adversary finds a $\hat{q}$ with constant probability.  
\end{proof}

\iffull
We can also show that the Swapping ACF is not adaptive; this proof is essentially identical to that of Theorem~\ref{thm:cycliclower}.

\begin{theorem}
  \label{thm:swappinglower}
  Let $\mathcal{F}$ be a Swapping ACF with $N = \Theta(n)$ slots.  
  Then there exists an adversarial strategy, making $\Theta(b/\epsilon^{b})$ queries, which wins the adaptivity game against $\mathcal{F}$ with probability $\Omega(1)$.  
\end{theorem}
\fi

One might think that hashing elements to another bucket (as is done in the Cuckooing ACF) is sufficient to make a filter adaptive.  
The reason Theorem~\ref{thm:cycliclower} gives such a strong lower bound for the Cyclic ACF is that when we move an element to the next fingerprint, it is still a false positive with probability $\epsilon$.  A constant number of these movements still leaves a significant probability that a false positive is not yet fixed.  In contrast, when a colliding element is moved in the Cuckooing ACF, it still collides with the query with probability only $\epsilon/n$---this seems low enough that almost all queries are successfully fixed after only a single movement.  

Nonetheless, the adversary can use a birthday attack to find a small set of elements that cannot all be simultaneously fixed.

We obtain lower bounds for a fairly broad class of filters, where the total total information stored (i.e. hash index plus location plus fingerprint) for each element is at most $\log(n/\epsilon) + O(1)$ bits.  This stands in contrast to the Broom Filter of Bender et al.~\cite{BenderFaGo18}, which is adaptive and which stores an \emph{average} of $\log(n/\epsilon) + O(1)$ bits---in short, this proof shows that the nonuniformity of hash lengths in~\cite{BenderFaGo18} is crucial to achieving adaptivity.

\begin{definition}
  \label{def:detkadaptive}
  A \emph{deterministic $k$-adaptive filter} $\mathcal F$ on $n$ elements with false positive rate $\epsilon$ is a filter satisfying the following:
  \begin{itemize}
  \item $\mathcal{F}$ has access to $k$ uniform random hash functions $h_0,\ldots,h_{k-1}$.  Each hash has length at most $\log (N/\epsilon)$, for some $N = O(n/k)$ with $N \geq n/k$. 
  \item For every configuration $C$ of $\mathcal{F}$, each $x_i\in S$ is stored using at least one hash $h_{C[i]}(x)$, $0\leq C[i] \leq k-1$.  
  \item The filter answers \textnormal{\present{}} to a query $q$ on configuration $C$ if there exists an $x_i$ such that $h_{C[i]}(q) = h_{C[i]}(x_i)$.  Otherwise, it answers \textnormal{\absent{}}.
  \item  On a false positive $q$, $\mathcal F$ updates $C$ to a new configuration $C'$ in round-robin order.  In particular, if a query $q$ collides with an element $x_i\in S$ stored using $h_{\beta}$, then $x_i$ is stored in $C'$ using $h_{\beta'}$ satisfying $\beta' = \beta+ 1 \pmod k$. 
\end{itemize}
\end{definition}

By setting each hash $h_i$ in Definition~\ref{def:detkadaptive} so that for any $i\in\{0,\ldots,k-1\}$ and $x\in U$, $h_i(x)$ is the concatenation of $h_i^\ell(x)$ and $h^f(x,i)$, 
the Cuckooing ACF is a deterministic $k$-adaptive filter. 
By setting $h_{(i,\alpha)}(x)$ to be the concatenation of $h_i^{\ell}(x)$ and $h^f(x,\alpha)$, the Cyclic ACF is a deterministic $k2^s$-adaptive filter.\footnote{Strictly speaking the Cyclic ACF does not satisfy Definition~\ref{def:detkadaptive} because the hashes are not all independent.  However, this only increases the ability of an adversary to find false positives.}

The round-robin ordering requirement stands out as being a bit artificial, but our proof can fairly easily be generalized to handle other deterministic methods to update the configuration.

\begin{theorem}
  \label{thm:detlower}
  There exists an adversarial strategy making $O(n)$ queries such that, for 
  any deterministic $k$-adaptive filter $\mathcal{F}$
  with $k < \log n/6\log \log n$ and $\epsilon > 1/(n^{1/k})$,
  the adversary wins the adaptivity game with probability $\Omega(1)$.  
\end{theorem}

The rest of this section proves Theorem~\ref{thm:detlower}.

\paragraph{Querying to Find a Mutually Unfixable Set.}
The proof of 
Theorem~\ref{thm:detlower} begins
with the adversary searching for a structure that ``blocks'' the filter, preventing it from fixing a false positive.

Consider a stored element $x_i \in S$, and fix a filter $\mathcal{F}$ with $k$ hash functions $h_0,\ldots h_{k-1}$.  A set of queries $K$ is called \defn{mutually unfixable for $x_i$} if
\begin{itemize}
\item for all $\beta\in \{0,\ldots,k-1\}$, there exists a $q'\in K$ 
  such that $h_\beta(x_i) = h_{\beta}(q')$, and 
\item  for all $q'\in K$ there exists a $\beta,\in \{0,\ldots,k-1\}$ such that $h_\beta(x_i) = h_{\beta}(q')$.  
\end{itemize}

The goal of our adversary is to find a mutually unfixable subset of the queries, as in any filter configuration, at least one element in a mutually unfixable set is a false positive.   

The adversary begins by 
choosing a set $Q$
of $(1 + 1/k)N/(\epsilon n^{1/k})$
queries selected uniformly at random from $U$.  
We show 
that if $Q$ is this size, then with constant probability $Q$ will contain $\Theta(1)$ mutually unfixable sets, each of size $O(k)$.  

The adversary then queries members of $Q$ for $2k$ rounds; any query that is a false positive during the second set of $k$ rounds is stored in a set $Q_d$.  
We show that $Q_d$ will be the union of some mutually unfixable subsets of $Q$.

Finally, the adversary repeatedly selects $k$ elements from $Q_d$ and randomly selects a permutation $P$ on these elements.  The adversary queries these elements in order, twice.  We show 
that, over $O(n)$ total queries, the adversary will (with constant probability) pick $k$ elements corresponding to a mutually unfixable set, and query them in an order such that each is a false positive every time it is queried.  With this strategy, the adversary can find a false positive $\hat{q}$ with constant probability.  

\iffull

  For each $\beta\in \{0,\ldots, k-1\}$, let $R_Q(\beta)$ be the set of values hashed to by $Q$ via $h_\beta$.  In other words, $y\in R_Q(\beta)$ when there exists $q\in Q$ 
  such that $h_{\beta}(q) = y$.   We begin by showing that for all $\beta$, $|R_Q(\beta)|$ is very close to $|Q|$ with constant probability.

\begin{lemma}
  \label{lem:sizeofrange}
  Let $Q$ be any set of elements of size $(1 + 1/k)N/(\epsilon n^{1/k})$, and assume $k \leq \log n/6\log\log n$.
  Then with probability at least $1/2$, $\min_{1\leq \beta<k} |R_Q(\beta)| > |Q|/(1 + 1/k)$.

\end{lemma}
\begin{proof}
  We say that a $q\in Q$ is 
  \defn{redundant} for a hash $h_\beta$ if there is a $q'\in Q\setminus\{q\}$ 
  such that $h_{\beta}(q) = h_{\beta}(q')$.  Then $|Q|-|R_Q(\beta)|$ is upper bounded by the number of redundant $q$ for $h_\beta$.

  Via a union bound, 
  the probability of a given query $q\in Q$ 
  hashing under a given $h_{\beta}$ to one of the $\leq k|Q|$ values hashed to by another $q'\in Q\setminus\{q\}$ is at most 
  $\epsilon |Q|/N$.  
  Thus, we can upper bound the number of redundant $q$ for $h_{\beta}$ by the sum of $|Q|$ Bernoulli trials, each of which succeeds with probability $\epsilon|Q|/N$.  
  Then the total number of redundant $(q,h_{\beta})$ in expectation is  $\epsilon |Q|^2/N$.
  Using Chernoff bounds~\cite[Exercise 4.7]{MitzenmacherUpfal17}, since  $\epsilon |Q|^2/N \geq N/n^{1/k} = \Omega(\log n)$, the probability that there are more than $2\epsilon |Q|^2/N$ redundant pairs is at most $1/n$.  Taking the union bound over all $\beta\in \{0,\ldots,k-1\}$, all $R(\beta)$ have at most $2\epsilon |Q|^2/N$ redundant pairs with probability $k/n$.

  Plugging into our bound for $|Q| - |R_Q(\beta)|$, we obtain that 
  $|R_Q(\beta)| > |Q|(1 - 2\epsilon |Q|/N)$ for all $\beta$ with probability $\geq 1 - k/n \geq 1/2$. (We weaken this bound to simplify.) Since 
  $k \leq \log n/6\log\log n$ implies $n^{1/k} > 2(k+1)/(1 + 1/k)$, we have $|Q| < N/(2\epsilon(k+1))$; and therefore $1 - 2\epsilon|Q|/N > 1/(1 + 1/k)$.
\end{proof}

  We observe that there exists a mutually unfixable set $K\subseteq Q$ for some $x_i\in S$ if and only if, for all $\beta\in\{0,\ldots, k-1\}$, element $x_i$ hashes to an element of $R_Q(\beta)$ under $h_{\beta}$.
  Let $X_Q$ consist of all $x_j\in S$ such that there is a $K^j\subseteq Q$ that is mutually unfixable for $x_j$.
  That is, $X_Q = \{x_j\in S ~|~ h_\beta(x_j)\in R_Q(\beta) \text{ for all } \beta\in \{0,\ldots, k-1\}\}$. 
  We begin by showing that $X_Q$ is nonempty with constant probability. 

\begin{lemma}
  \label{lem:probmutunfixset}
  For any deterministic $k$-adaptive 
  filter $\mathcal F$ storing a set of elements $S$ of size $n$, and for
  any set of $(1 + 1/k)N/(\epsilon n^{1/k})$ query elements $Q$,
  $Q$ contains a  $K\subseteq Q$ such that $K$ is mutually unfixable for some $x_j\in S$. 
  with probability $> (1 - 1/e)/2$.  
\end{lemma}
\begin{proof} 
  By independence of the hashes, the probability that a given $x$ has a mutually unfixable set is $\prod_{\beta=1}^k (\epsilon |R_Q(\beta)|/N)$.   By Lemma~\ref{lem:sizeofrange}, 
  with probability $1- 1/n$, this is at least $\left(\frac{\epsilon |Q|}{N(1 + 1/k)}\right)^k$.  Substituting for $|Q|$, this means that each $x_i\in S$ has a mutually unfixable set with probability at least $1/2$. 

  Thus, if Lemma~\ref{lem:sizeofrange} is satisfied, the probability that \emph{no} $x_i$ has a mutually unfixable set is $(1 - 1/n)^n \leq 1/e$.   Therefore, $Q$ has a mutually unfixable set for some $x$ with probability at least $(1/2)(1 - 1/e)$.
\end{proof}

The next lemma upper bounds $|X_Q|$; with the above we have $|X_Q| = \Theta(1)$ with constant probability.

\begin{lemma}
  \label{lem:smallnummutunfixsets}
For any deterministic $k$-adaptive 
  filter $\mathcal F$ storing a set of elements $S$ of size $n$, and
  any set of $(1 + 1/k)N/(\epsilon n^{1/k})$ query elements $Q$, let $X_Q$ be the elements of $S$ with a mutually unfixable subset of $Q$.  
  Then with constant probability, $|X_Q| = O(1)$.  
 
\end{lemma}
\begin{proof}
  We have $|R_Q(\beta)| \leq |Q|$ for all $0\leq\beta\leq k-1$ by definition; thus, the probability that a given $x$ has a mutually unfixable subset is $(\epsilon|R_Q|/N)^k \leq e/n$.  Thus, $e$ elements $x_i\in S$ have mutually unfixable subsets of $Q$ in expectation; by Markov's inequality there are $O(1)$ elements in $S$ with mutually unfixable subsets with constant probability.  Each maximal mutually unfixable set must be unfixable for some $x_i\in S$, so this proves $u = O(1)$.
\end{proof}

Finally, we bound the size of each mutually unfixable set.

  \begin{lemma}
    \label{lem:sizemutunfixablesets}
    If $k < \log n/6\log\log n$, then 
    with constant probability, for all mutually unfixable $K\subseteq Q$, $k < |K| = O(k)$. 
  \end{lemma}

\begin{proof}
  Consider an element $x_i$ such that some $K^i\subseteq\nobreak Q$ of size $|K^i| = k$ is mutually unfixable for $x_i$.  A further element $q\in Q\setminus K^i$ collides with $x_i$ (on some hash) with probability at most $\epsilon/N$; summing over $|Q| < n$ elements we obtain $\epsilon (n/N) < O(k)$ additional elements in expectation.  By Markov's inequality, there are $O(k)$ such elements with constant probability. 
From Lemma~\ref{lem:smallnummutunfixsets}, there are $O(1)$ elements $x\in S$ with mutually unfixable sets with constant probability.  Therefore, this bound applies to all $O(1)$ mutually unfixable sets with constant probability.

  To lower bound the size of these sets, we have that if the size of a  set $K^j$ that is mutually unfixable for some $x_j\in S$ is less than $k$, there exists a $q_1\in K^j$ and two hash indices $\beta_1,\beta_2\in\{0,\ldots,k-1\}$ such that $q_1$ collides with $x_j$ under both $h_{\beta_1}$ and $h_{\beta_2}$.  The probability that there exists such a pair $(q_1,x_j)$ is at most $k^2|Q|/n$, which is much less than a constant since $k < \log n/6\log\log n$.  Thus, $|K| \geq k$ for all mutually unfixable $K$ with constant probability.  
\end{proof}

\subsection{Lower bound for deterministic k-adaptive filters}
\label{sec:detlower}

The adversary begins by selecting a set of 
$(1 + 1/k)N/(\epsilon n^{1/k})$ random queries $Q$.  
The adversary performs $2k$ rounds, where in each round, the adversary queries every element of $Q$ in random order.  We call the first $k$ such rounds the \defn{preliminary rounds}, and we call the second $k$ rounds the \defn{testing rounds}.

Let $Q_d$ be the set of elements that are false positives during any of the testing rounds.  
If $|Q_d| \neq O(k)$, the adversary fails.\footnote{Specifically, for a sufficiently large constant $c$, the adversary fails if $|Q_d| > ck$.  The constant $c$ gives a tradeoff between the lower-order terms in the number of queries performed, and the constant probability that the adversary wins.}

The adversary performs a further $\binom{|Q_d|}{k} k!$ rounds, the \defn{permutation rounds.}  In each permutation round, 
the adversary picks a random subset of $Q_d$ of size $k$, which we call $Q_r$. 
The adversary then picks a random permutation $P$ of $Q_r$.  
The adversary queries all elements in $Q_r$ in the order given by $P$,  then repeats: again querying all elements in $Q_r$ in the order given by $P$.  
If every query in these repeated queries is a false positive, 
the adversary stops, and sets $\hat{q}$ to be the first element in $P$.

We now analyze this adversary.  
The idea is that after the $k$ preliminary rounds, the false positives during the testing rounds will consist of a union of mutually unfixable subsets of $Q$.  We prove this in the following lemma.

\begin{lemma}
  \label{lem:detfps}
  For any $Q_d$ obtained by the adversary, there exists a sequence of $u$ sets $K_1,\ldots, K_u \subseteq Q$ such that 
  for all $\lambda$, $K_\lambda$ is a mutually unfixable set for some element $x_i\in S$, and 
  $Q_d = \bigcup_{\lambda=1}^u K_\lambda$.
\end{lemma}
\begin{proof}
First, we show that every $q\in Q_d$ is a member of a mutually unfixable subset $K$.  
  Consider an element $x_i\in S$ that does not have a mutually unfixable $K^i\subseteq Q$.  
  In each round where $x_i$ collides with a false positive of $Q$, $x_i$ must be moved to the next hash.  After less than $k$ movements, $x_i$ is moved to some hash $h_{\beta}$ where $x_i$ does not collide with any element of $Q$ under $\beta$; that is, $h_\beta(x_i)\notin R_Q(\beta)$.  After this, $x_i$ will not be moved during subsequent queries from $Q$, and will not cause any elements of $Q$ to be a false positive.

  Thus, any $x_i$ without a mutually unfixable $K^i\subseteq Q$ will not collide with any query in the testing rounds.  This means that all $q'\in Q_d$ collide with an $x_j$ that has a mutually unfixable subset $K^i\subseteq Q$.  By definition, $q'\in K^i$.

  Now, consider a maximal mutually unfixable set $K^\ell\subseteq Q$, which is mutually unfixable for some $x_\ell\in S$.  We show that $Q_d$ contains a mutually unfixable $K'\subseteq K^\ell$.  During each round, at least one member of $K^\ell$ collides with $x_\ell$ by definition; thus $x_\ell$ is moved at least $k$ times during the testing rounds.  This means that for each $\beta\in \{0,\ldots, k-1\}$, $x_\ell$ collides with a $q\in Q$ while being stored using hash $h_{\beta}$.  Thus, by definition, the set of elements that collide with $x_\ell$ during the testing rounds form a mutually unfixable $K'\subseteq K^\ell$.  
\end{proof}

With the above, we show that the adversary wins the adaptivity game with constant probability. 

\begin{lemma}
  \label{lem:adversarywins}
  With probability $\Omega(1)$, the adversary finds a $\hat{q}$ which is a false positive during the permutation rounds.  
\end{lemma}

\begin{proof}
	We assume that Lemma~\ref{lem:smallnummutunfixsets} is satisfied; this happens with probability $\Omega(1)$.  Combining Lemma~\ref{lem:smallnummutunfixsets} with Lemma~\ref{lem:detfps}, $Q_d$ consists of $u = O(1)$ mutually unfixable sets $K_1',\ldots K_u'$, each of size $O(k)$.  

  Let $K^i$ be a mutually unfixable set of an element $x_i\in S$ such that $K^i \subseteq Q_d$ and $|K^i| = k$.
  We analyze the probability that, in a given permutation round, the adversary chooses $Q_r = K^i$, and chooses a permutation $P$ that causes each query to be a false positive; this means that the adversary selects a $\hat{q}$ with probability $\Omega(1)$. Then we show that if the adversary selects a $\hat{q}$, it is a false positive with probability $\Omega(1)$.

  Fix a permutation round, and let $\beta_i$ be the hash index of $x_i$ at the beginning of this permutation round.
  Let $P_{K^i} = (q_1, q_2, \ldots, q_k)$ be a permutation on $K^i$ such that for each $q_{j}$, $h_{\beta_i + j\pmod k}(q_{j}) = h_{\beta_i + j \pmod k}(x_i)$. 

  If $P_{K^i}$ is selected, then all queries to $P_{K^i}$ will be false positives, and the adversary wins the adaptivity game with probability $1$.  We select $Q_r$ to consist of the members of $K^i$ with probability $1/\binom{|Q_d|}{k}$, and we query them in order of $P_{K^x}$ with probability $1/(k!)$.  Repeating this process for $\binom{|Q_d|}{k}k!$ rounds gives constant probability of success.

  Now, we show that, with constant probability, the adversary only sets $\hat{q}$ for a permutation $P' = (q_1',q_2',\ldots,q_k')$ if $P'$ satisfies, for some $x_\ell$ with hash index $\beta_\ell$,  $h_{\beta_\ell + j\pmod k}(q_{j}) = h_{\beta_\ell + j\pmod k} (x_\ell)$ for all $1\leq j \leq k'$.
  Recall that the adversary sets $\hat{q}$ when the $j$th query to $P$ is a false positive for all $j$; thus for each $q_j$ there exists an $x_j\in S$ and a $0\leq \beta_j < k$ such that $h_{\beta_j}(q_j) = h_{\beta_j}(x_j)$.  
  If $x_j = x_\ell$ for all $j$ then we are done.

  Otherwise, let $\hat{X}\subseteq S\setminus \{x_\ell\}$ be the elements of $S$ (other than $x_i$) that collide with an element of $Q_r$.  If all  $x_j\in \hat{X}$ only collide with a member of $Q_r$ under a single hash function $h_{\beta}$ (i.e. there do not exist $q_1,q_2\in Q_r$ and distinct $\beta_1,\beta_2\in \{0,\ldots,k-1\}$ with $h_{\beta_1}(x_j) = h_{\beta_1}(q_1)$ and $h_{\beta_2}({x_j}) = h_{\beta_2}(q_2)$), then ${x_j}$ will only be a false positive once.  Thus, all collisions in the repeated queries will collide with $x_i$, and again we are done.

  The probability that there exists an ${x_\ell}\in S$ and distinct $\beta_1,\beta_2\in \{0,\ldots, k-1\}$ with $h_{\beta_1}(x_\ell) = h_{\beta_1} (q_1)$ and $h_{\beta_2}(x_\ell) = h_{\beta_2} (q_2)$ for some $q_1, q_2\in Q_d$ is at most $k^2|Q_d|^2\frac{\epsilon^2}{n^2} = o(1)$.
\end{proof}

The number of queries required for this adversary is, for some constant $c$ depending on the maximum size of $Q_d$, $O(n^{1-1/k}/\epsilon  + c^k k!)$; this is $O(n)$ since $k < \log n/6\log \log n$ and $\epsilon > 1/(n^{1/k})$.  This completes the proof of Theorem~\ref{thm:detlower}.

\fi

\bibliographystyle{plain}
\bibliography{bloom}
\newpage
\appendix
\section{Symbol Table}
\label{sec:symboltable}

The following symbols are used throughout the paper: 

\begin{center}
\begin{tabular}{|c|c|}
\hline
Symbol & Usage \\
\hline
\hline
$S$ & The set of elements stored in the filter \\
\hline
$x_i$ & The $i$th element stored in the filter\\
\hline
$\epsilon$ & The static false positive rate; used as a parameter in adaptive filters\\
\hline
$U$ & Universe of possible elements; all queries and stored elements are from $U$\\
\hline
$Q$ & The set of elements queried in a query sequence\\
\hline
$h^\ell_i$ & The $i$th location hash function\\
\hline
$h^f_i$ & The $i$th fingerprint hash function\\
\hline
$f$ & The size of the fingerprint\\
\hline
$k$ & The number of hash tables\\
\hline
$b$ & The number of slots per bin\\
\hline
$N$ & The number of bins in a hash table; $N = \gamma n/(bk)$\\
\hline
$\beta_i$ & The hash index of $x_i$ (i.e. the table storing $x_i$)\\
\hline
$B(x_i)$ & The bin storing the fingerprint of $x_i$;  $B(x_i) = h_{\beta_i}^\ell(x_i)$\\
\hline
$C$ & A configuration of a filter, defined using the hash index of each stored element $x_i$\\
\hline
$C[i]$ & The hash index of $x_i$ in configuration $C$\\
\hline
$s$ & The number of hash selector bits in a Cyclic ACF\\
\hline
$\alpha$ & The value stored using the hash selector bits for a given slot\\
\hline
$A$ & The set of elements queried during a given experiment\\
\hline
\end{tabular}
\end{center}

The following symbols are used in the proof that the Cuckooing ACF is adaptive in Section~\ref{sec:upper}:
\begin{center}
\begin{tabular}{|c|c|}
\hline
Symbol & Usage \\
\hline
\hline
$C_0$ & The configuration before any query is performed\\
\hline
$C_i$ & The configuration after query $q_i$\\
\hline
$k_i$ & The number of elements moved when fixing $q_i$ on configuration $C_{i-1}$\\
\hline
$x_{i_1},x_{i_2},\ldots x_{i_{k_i}}$ & The sequence of elements moved when fixing $q_i$ on configuration $C_{i-1}$\\
\hline
	$B(i,C)$ & The bin storing $x_i$ in configuration $C$; $B(i,C) = h^{\ell}_{C[i]}(x_i)$\\
\hline
	$B'(i,C)$ & The alternate bin that can store $x_i$ in configuration $C$; $B'(i,C) = h^{\ell}_{1-C[i]}(x_i)$\\
\hline
$F_0$ & The set of initial false positives\\
\hline
$k_i^0$ & The number of elements moved when fixing $q_i$ on configuration $C_0$\\
\hline
$x_{i_1'},x_{i_2'},\ldots x_{i_{k_i^0}'}$ & The sequence of elements moved when fixing $q_i$ on configuration $C_0$\\
\hline
$\Phi(t)$ & The potential after query $t$\\
\hline
\end{tabular}
\end{center}

The following symbols are used in the lower bound (the proof of Theorem~\ref{thm:detlower}) in Section~\ref{sec:lower}:
\begin{center}
\begin{tabular}{|c|c|}
\hline
Symbol & Usage \\
\hline
\hline
$\hat{q}$ & A special query that the adversary names as a likely false positive\\
\hline
$R_Q(\beta)$ & The set of values that can be obtained by applying hash $h_\beta$ to any element in $Q$\\
\hline
$X_Q$ & The set of all $x_j$ such that there exists a $K^j\subseteq Q$ that is mutually unfixable for $x_j$ \\
\hline
$Q_d$ & The set of queries  that are false positives during any of the testing rounds\\
\hline
$Q_r$ & A random subset of $Q_d$ of sized $k$; this subset is queried during a permutation round\\
\hline
$P$ & A permutation applied to some $Q_r$ during a permutation round\\
\hline
$P_{K^i}$ & A permutation on a mutually unfixable set of $x_i$ with certain hash collision properties\\
\hline
$\hat{X}$ &  The elements of $S$ other than a fixed $x_i$ that collide with any element of $Q_r$\\
\hline
\end{tabular}
\end{center}

\newpage
\section{Experiments On Network Traces}
\label{sec:networkexperiments}

\begin{figure}[ht]
  \includegraphics[width=.46\textwidth]{plots/8bitChicagoA.pdf}
  \includegraphics[width=.46\textwidth]{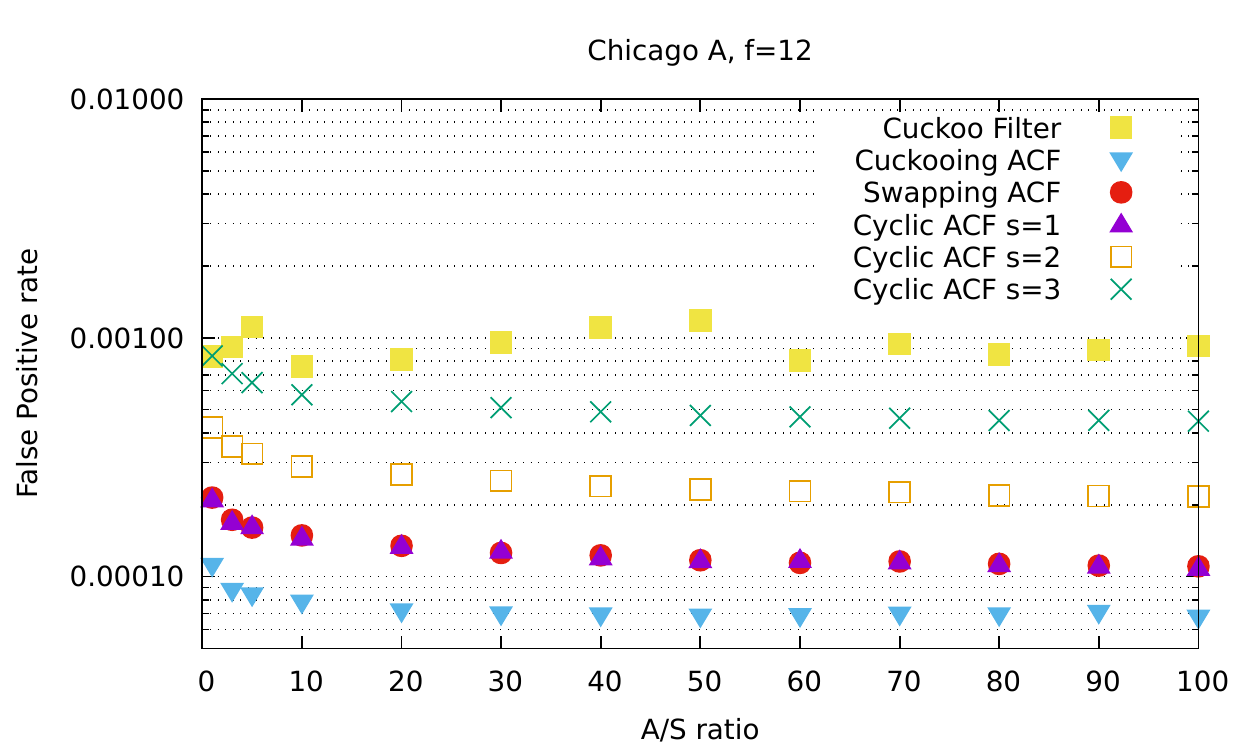}
  \includegraphics[width=.46\textwidth]{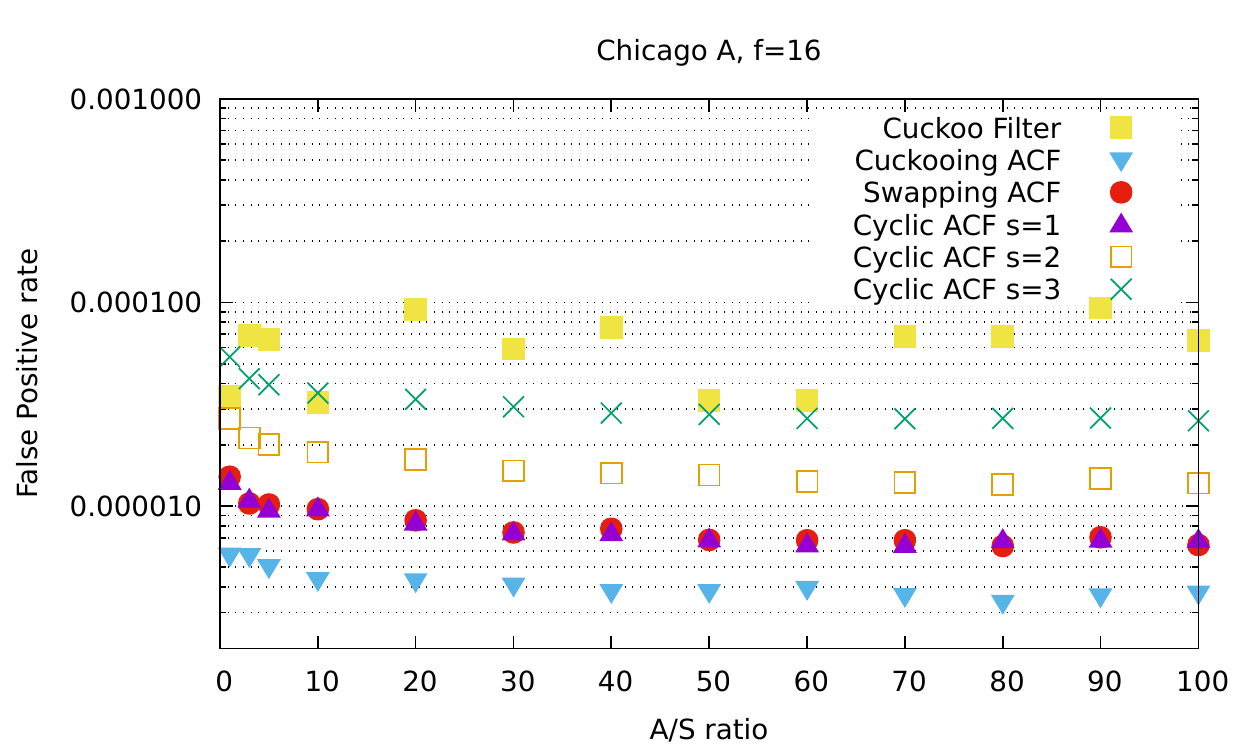}
  \caption{The false positive rate incurred for each filter on the Chicago A dataset, with $8$, $12$, and $16$ fingerprint bits.}
\end{figure}

\begin{figure}[ht]
  \includegraphics[width=.46\textwidth]{plots/8bitChicagoB.pdf}
  \includegraphics[width=.46\textwidth]{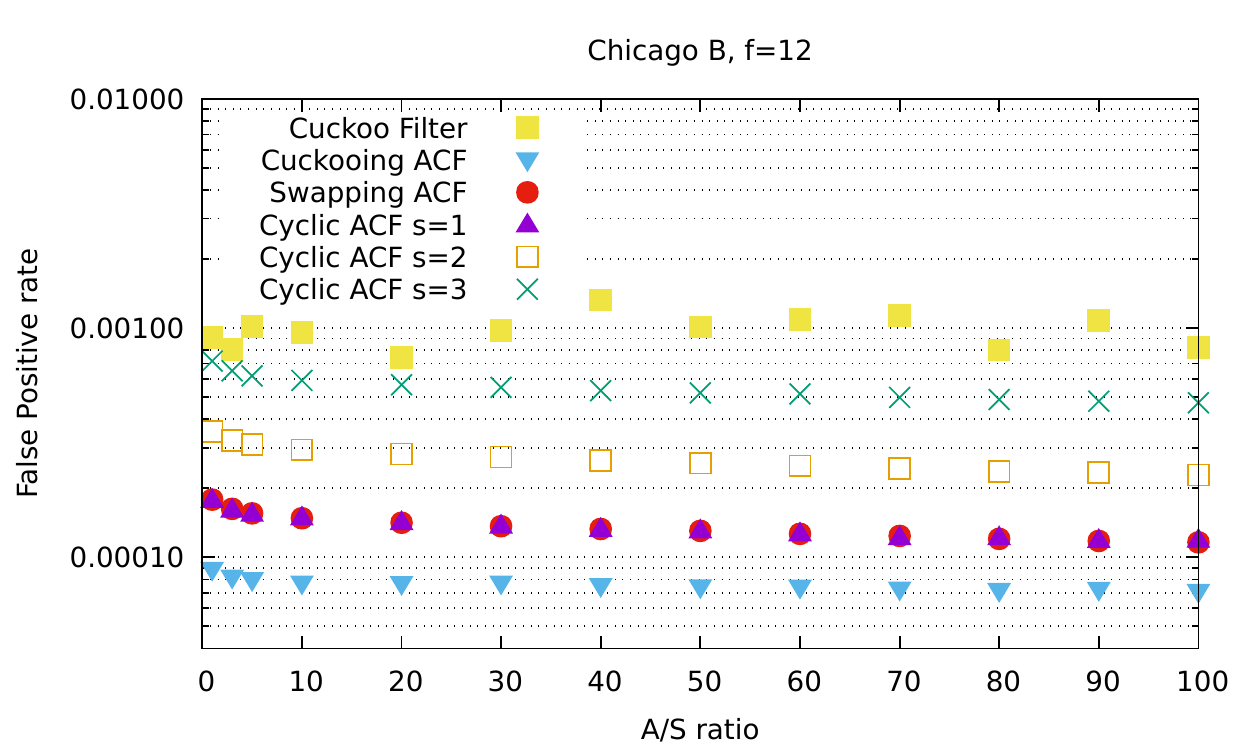}
  \includegraphics[width=.46\textwidth]{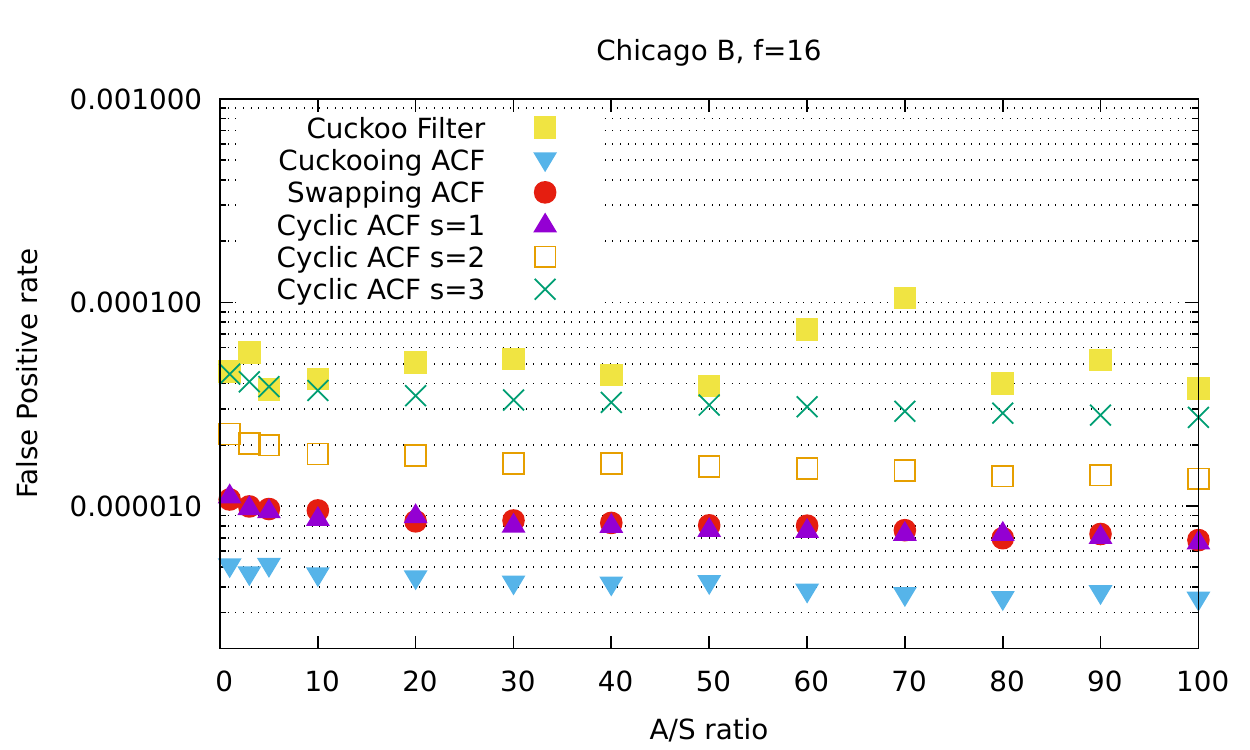}
  \caption{The false positive rate incurred for each filter on the Chicago B dataset, with $8$, $12$, and $16$ fingerprint bits.}
  \vspace{-.2in}
\end{figure}

\begin{figure}[ht]
  \includegraphics[width=.46\textwidth]{plots/8bitSanjose.pdf}
  \includegraphics[width=.46\textwidth]{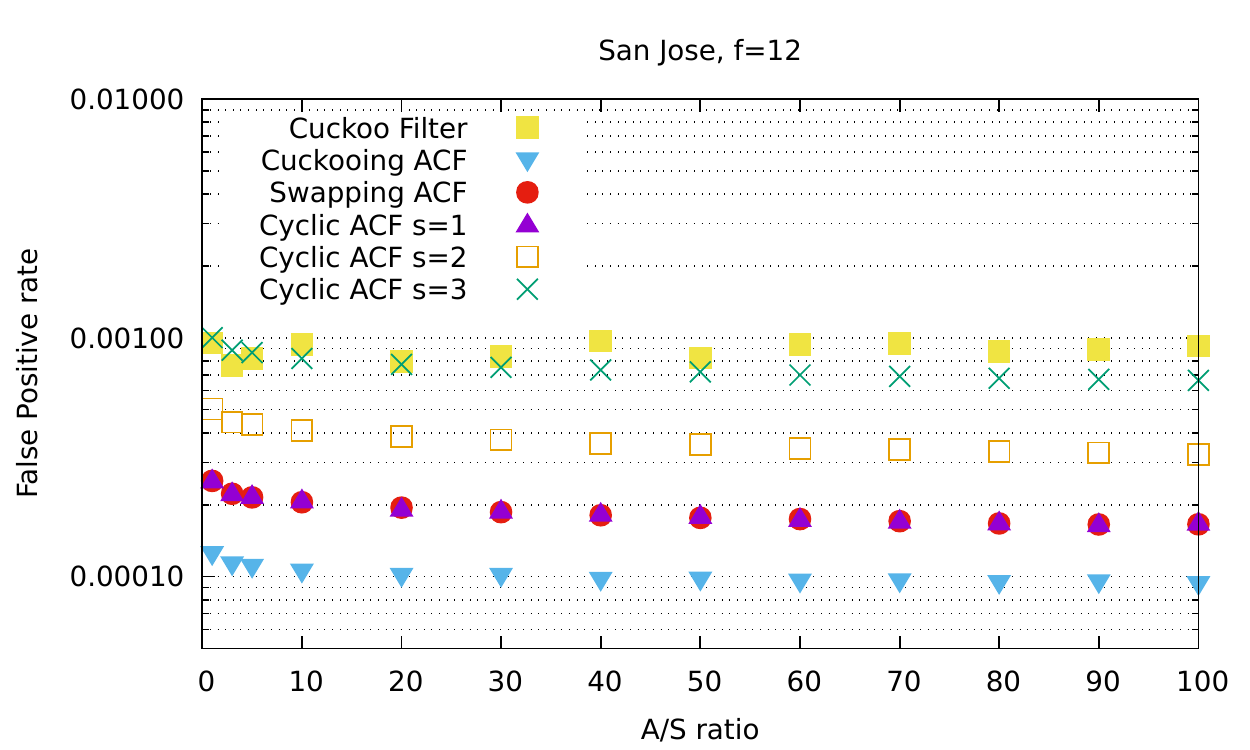}
  \includegraphics[width=.46\textwidth]{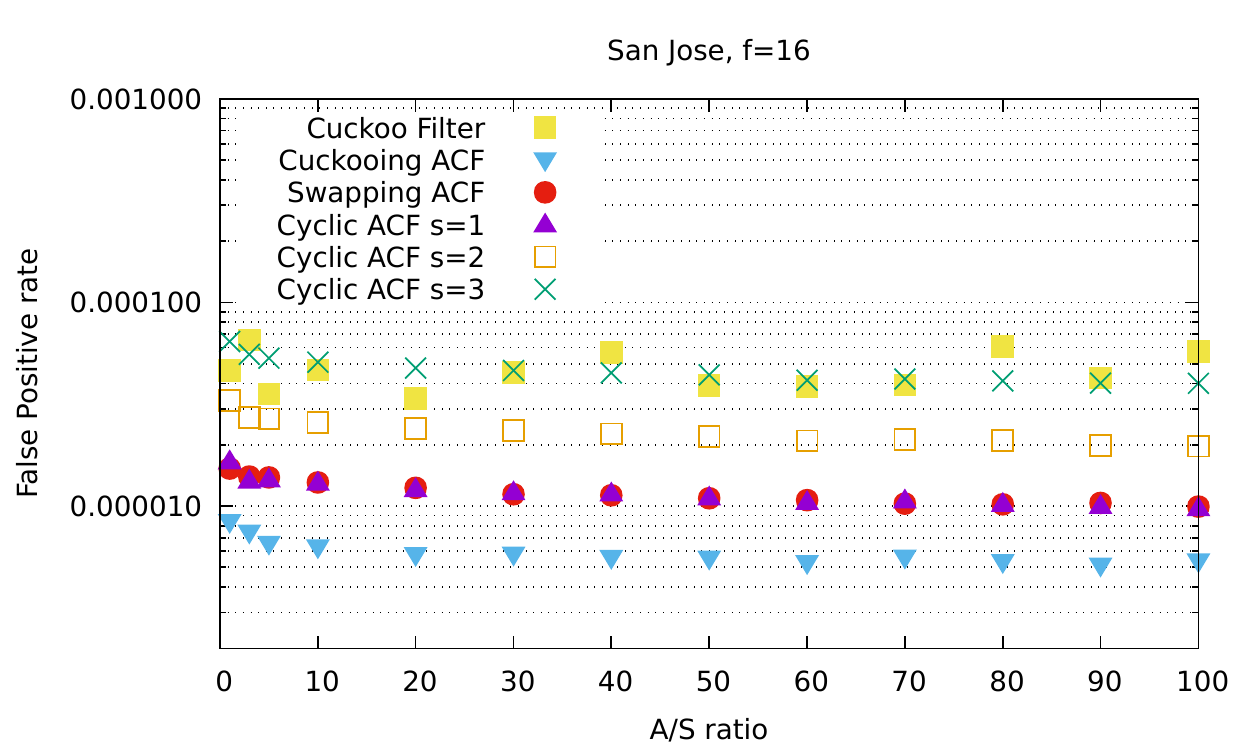}
  \caption{The false positive rate incurred for each filter on the San Jose dataset, with $8$, $12$, and $16$ fingerprint bits.}
\end{figure}
 
\end{document}